\documentclass[11pt]{article}

\usepackage{amsmath,amssymb,graphicx, mathtools}

\usepackage{graphics,graphicx,color}         	
\usepackage{epstopdf}  
\usepackage{url}		
\usepackage{hyperref}
\usepackage{cleveref}
\usepackage{amsthm}
\usepackage[margin= 1in]{geometry}

\usepackage{floatrow,subcaption}
\newfloatcommand{capbtabbox}{table}[][0.5\textwidth]

\usepackage{color}  

\usepackage{fancyhdr}

\usepackage[boxed]{algorithm2e} 

\newtheorem{prop}{Proposition}

\newcommand{\xcondh}{\hat{\boldsymbol{x}}_{\text{cond}}(\boldsymbol\theta)}
\newcommand{\condh}{\widehat{\boldsymbol{\Gamma}}_{\text{cond}}(\boldsymbol\theta)}
\newcommand{\condhinv}{\widehat{\boldsymbol{\Gamma}}^{-1}_{\text{cond}}(\boldsymbol\theta)}
\newcommand{\xcond}{\boldsymbol{x}_{\text{cond}}(\boldsymbol\theta)}
\newcommand{\cond}{\boldsymbol{\Gamma}_{\text{cond}}(\boldsymbol\theta)}
\newcommand{\condinv}{\boldsymbol{\Gamma}_{\text{cond}}^{-1}(\boldsymbol\theta)}

\newcommand{\N}{\mathcal{N}}
\newcommand{\prior}{\B{\Gamma}_{\rm prior}}
\newcommand{\eqdef}{\stackrel{def}{=}}

\newcommand{\B}{\boldsymbol}
\newcommand{\ip}[2]{\langle {#1}, {#2} \rangle}
\newcommand{\priorcov}{\B{\Gamma}_{\text{prior}}}
\newcommand{\bb}{\mathbf b}
\newcommand{\dx}{\mathrm{d}}
\newcommand{\mc}[1]{\mathcal{#1}}
\newcommand{\mb}[1]{\mathbb{#1}}

\Crefname{prop}{Proposition}{Propositions}
\Crefname{algocf}{Algorithm}{Algorithms}

\title{Efficient Marginalization-based MCMC Methods for Hierarchical Bayesian Inverse Problems}

\author{Arvind K. Saibaba\thanks{Department of Mathematics, North Carolina State University, Raleigh, NC
    asaibab@ncsu.edu} \and
 Johnathan Bardsley\thanks{Department of Mathematical Sciences, University of Montana, Missoula, MT bardsleyj@mso.umt.edu} \and
 D. Andrew Brown\thanks{School of Mathematical and Statistical Sciences, Clemson University, Clemson, SC ab7@clemson.edu} \and
Alen Alexanderian\thanks{Department of Mathematics, North Carolina State University, Raleigh, NC alexanderian@ncsu.edu}
}  

\begin{document}
\maketitle

\begin{abstract}
Hierarchical models in Bayesian inverse problems are characterized by an assumed prior probability distribution for the unknown state and measurement error precision, and hyper-priors for the prior parameters. Combining these probability models using Bayes' law often yields a posterior distribution that cannot be sampled from directly, even for a linear model with Gaussian measurement error and Gaussian prior, {both of which we assume in this paper. In such cases, Gibbs sampling can be used to sample from the posterior \cite{Bar11}}, but problems arise when the dimension of the state is large. This is because the Gaussian sample {required for} each iteration can be prohibitively expensive to compute, and because the statistical efficiency of the Markov chain degrades as the dimension of the state increases. The latter problem can be mitigated using marginalization-based techniques, such as those found in \cite{FoxNor,joyce2018point,RueHel}, but these can be computationally prohibitive as well. In this paper, we combine the low-rank techniques of \cite{brown2016computational} with the marginalization approach of \cite{RueHel}. We consider two variants of this approach: delayed acceptance and pseudo-marginalization. We provide a detailed analysis of the acceptance rates and computational costs associated with our proposed algorithms, and compare their performances on two numerical test cases---image deblurring and inverse heat equation.\end{abstract}

\section{Introduction}
Inverse problems arise in a wide range of applications and typically involve
estimating unknown parameters in a physical model from noisy, indirect
measurements.
{In the applications that we are interested in, the unknown parameter
vector results from the discretization of a continuous function defined on the
computational domain and hence is high-dimensional. Additionally, the inverse
problem of recovering the unknown parameters from the data is
\textit{ill-posed}: a solution may not exist, may not be unique, or may be
sensitive to the noise in the data.}

To set the context of our work and some notation, we consider a discrete
measurement error model of the form
\begin{equation}
\label{Ax=b}
\B{b}=\B{A}\B{x}+\B{\epsilon},\quad \B{\epsilon}\sim\N(\B{0},\mu^{-1}\B{I}_M),
\end{equation}
where $\B{b}\in\mathbb{R}^M$ denotes the measurements, $\B{A}\in\mathbb{R}^{M\times N}$ is the discretized forward model, $\B{x}\in\mathbb{R}^N$ is the unknown estimand, and $\mu > 0$ is the error precision. The inverse problem here seeks to recover the unknown $\B{x}$ from the measurement $\B{b}$.

The statistical model~\cref{Ax=b} implies that
the probability density function for $\B{b}$ is given by
\begin{equation}
\label{data probability}
\pi(\B{b}\mid\B{x},\mu)\propto {\mu}^{M/2}\exp\left(-\frac{\mu}{2}\Vert\B{A}\B{x}-\B{b}\Vert^2_2\right),
\end{equation}
where `$\propto$' denotes proportionality.
To address the ill-posedness, we also assume a Gaussian prior probability density function (akin to choosing a quadratic regularization function) of the form
\begin{equation}
\label{prior}
\pi(\B{x}\mid\sigma)\propto \sigma^{N/2}\exp\left(-\frac\sigma2\B{x}^\top\prior^{-1}\B{x}\right).
\end{equation}
{For convenience, we take the prior to have zero mean, but a nonzero mean can also be easily incorporated into our framework.}
Then, if we define $\B{\theta}=(\mu,\sigma)$, through {\em Bayes' law} we
obtain the {\em posterior density function} of $\B{x}$ conditioned on $\B{b}$
and $\B{\theta}$:
\begin{equation}\label{BayesLaw}
\begin{aligned}
\pi(\B{x}\mid\B{b},\B{\theta}) \propto& \> \pi(\B{b}\mid\B{x},\mu)\pi(\B{x}\mid\sigma)\\
 \propto& \>\mu^{M/2}\sigma^{N/2}\exp\left(-\frac\mu2\Vert\B{A}\B{x}-\B{b}\Vert^2_2-\frac\sigma2\B{x}^\top\prior^{-1}\B{x}\right).
\end{aligned}
\end{equation}
The maximizer of $\pi(\B{x}\mid\B{b},\B{\theta})$ is known as the maximum a
posteriori (MAP) estimator, which is also the minimizer of $-\ln
\pi(\B{x}\mid\B{b},\B{\theta})$. This, in turn, has the form of a Tikhonov
regularization, thus establishing a connection between Bayesian and classical
inverse problems.

The scaling terms in~\cref{BayesLaw} involving $\mu$ and $\sigma$ arise from
the normalizing constant of the Gaussian measurement error and prior models. In
the hierarchical Bayesian approach we will also treat $\B{\theta}=(\mu,\sigma)$
as unknown. The {\em a priori} uncertainty about plausible values of $\B\theta$
is quantified in the prior distribution with density $\pi(\B{\theta})$. For
convenience, we assume that the precision parameters $\mu$ and $\sigma$ are
independent and Gamma-distributed. Specifically, we assume
\begin{equation}
\label{hyperprior}
\pi(\B{\theta})=\pi(\mu)\pi(\sigma)\propto \mu^{\alpha_\mu-1}\sigma^{\alpha_\sigma-1}\exp(-\beta_\mu\mu-\beta_\sigma\sigma).\end{equation}
Other choices are also possible and are discussed in~\cite{Gelman06, ScottBerger06, brown2016computational} and elsewhere. Applying Bayes' law again yields a posterior density function of $\B{x}$ and $\B{\theta}$ conditioned on $\B{b}$:
\begin{equation}\label{eqn:post} \pi(\B{x},\B{\theta}\mid\B{b})\propto \pi(\B{x}\mid\B{b},\B{\theta})\pi(\B{\theta}).\end{equation}
The marginal distribution $\pi(\B{\theta}\mid \B{b})$ can be derived by
integrating over $\B{x}$; i.e.,  $\pi(\B{\theta}\mid \B{b}) =
\int_{\mathbb{R}^n}\pi(\B{x}, \B{\theta}\mid \B{b})d\B{x}$. Explicit
expressions for these distributions are provided in~\cref{sec:review}.

The focus of this paper is on the {problem of sampling from the posterior distribution} $\pi(\B{x},\B{\theta}\mid\B{b})$. In general, the full-posterior is not Gaussian and to explore this distribution the prevalent approach is to use Markov chain Monte Carlo (MCMC) methods \cite{metropolis1953equation,hastings1970monte,geman1993stochastic,tierney1994markov,GelfandSmith90}. { For a more comprehensive review of MCMC methods, please refer to \cite{brooks2011handbook}.} Several MCMC algorithms have been treated in the context of inverse problems in recent literature, which we briefly review. Specifically, in \cite{Bar11}, conjugacy relationships are exploited to define a Gibbs sampler, in which samples from the conditional densities $\pi(\B{\theta}\mid\B{b},\B{x})$ and $\pi(\B{x}\mid\B{b},\B{\theta})$ (which are Gamma- and Gaussian-distributed, respectively) are cyclically computed. { The computational cost of this Gibbs sampler is prohibitive for $N$ sufficiently large. This is due to the fact that as $N\rightarrow\infty$, the {\em integrated autocorrelation time} of the MCMC chain also tends to $\infty$, meaning that the number of Gibbs samples must increase with $N$; see \cite{AgaBarPapStu} for details. And second, computing samples from $\pi(\B{x}\mid\B{b},\B{\theta})$ requires the solution of an $N\times N$ linear system, and hence the computational cost of the individual Gibbs samples also increases with $N$.} 

{To address the first computational issue, i.e., that the correlation in the MCMC chain increases with $N$, an alternative to the Gibbs sampler is presented in \cite{RueHel}, where a proposed state $(\B{x}_*,\B{\theta}_*)$ is computed in two stages by first drawing $\B{\theta}_*$ from a proposal distribution and then drawing $\B{x}_* \sim \pi(\B{x} \mid \B{\theta}_*, \B{b})$.} The proposal $(\B{x}_*,\B{\theta}_*)$ is accepted or rejected jointly using a Metropolis-Hastings {step to obtain an approximate draw} from the posterior $\pi(\B{x},\B{\theta}\mid\B{b})$. This approach, called the {\em one-block} algorithm \cite{RueHel}, does not have the same degeneracy issues as the Gibbs sampler as $N\rightarrow\infty$. However, it can be expensive to implement when evaluating $\pi(\B{\theta}\mid\B{b})$ is computationally demanding and one still has to compute a sample from $\pi(\B{x}\mid\B\theta_\ast,\B{b})$ at every iteration.

{To address the second computational issue, i.e., that the cost of computing samples from $\pi(\B{x}\mid\B{b},\B{\theta})$  increases with $N$, we implement the approach taken in~\cite{brown2016computational}, where a low-rank approximation of the so-called prior preconditioned data misfit part of the Hessian is used.} This low-rank representation allows efficient sampling from the conditional distribution, reducing the overall computational cost. { When the forward operator is defined using partial differential equations, computing the conditional covariance matrix once may require hundreds of thousands of PDE solves; in the context of an MCMC algorithm which requires repeated computation of the conditional covariance, this can be prohibitively expensive. Even when the conditional covariance can be formed, storing it requires $\mc{O}(N^2)$ in memory and $\mc{O}(N^3)$ in computational cost, which is infeasible when $N$ is large (e.g., $O(10^5)$).} When the forward operator $\B{A}$ and the prior covariance matrix are diagonalized by the Fourier transform, algorithms such as the Fast Fourier Transform can be used effectively~\cite{bardsley2013efficient}. Other methods based on Krylov subspace solvers, e.g., Conjugate Gradient, have also been developed~\cite{gilavert2015efficient}. All of these approaches still suffer from the degeneracy issue as $N\rightarrow \infty$.

The contributions of this paper are as follows. {First, to tackle
the two drawbacks of the Gibbs sampler, as described above, we combine the
approaches of the previous two paragraphs, which to our knowledge has not been
done elsewhere. The use of a low-rank approximation defines an approximate
posterior density function $\hat\pi(\B{x},\B{\theta}\mid\B{b})$, whose samples
are only approximate and thus must be embedded within a
Metropolis-Hastings or importance sampling framework.} In our first algorithm,
which we call {\em approximate one-block}, $\hat\pi(\B{x},\B{\theta}\mid\B{b})$
is used as a proposal for Metropolis-Hastings, with the proposal samples
computed using the one-block algorithm. We propose two other variants of the
one-block algorithm that make use of $\hat\pi(\B{x},\B{\theta}\mid\B{b})$.
Specifically, we embed one-block applied to
$\hat\pi(\B{x},\B{\theta}\mid\B{b})$ within both the {\em delayed acceptance}
\cite{ChrFox} and {\em pseudo-marginal} \cite{PseudoMarg} frameworks to obtain
samples from the full posterior $\pi(\B{x},\B{\theta}\mid \B{b})$. {
Thus the algorithms we propose result from combing the low rank approximation
approach of \cite{brown2016computational} with one of the existing MCMC methods
mentioned above. To increase the novelty of our work, and also to provide the
user with some intuition on how well our algorithms can be expected to perform
in practice, we present theoretical results that provide insight into the
acceptance rates and the performance of the algorithms.} The main take away is
that when the low-rank approximation is sufficiently close (in a sense that we
make precise in~\ref{sec:analysis}), the algorithms have similar behavior to
the one-block algorithm.

{Hierarchical Bayesian approaches have been applied to inverse
problems in various other works, going back to \cite{KaiSom05} and more
recently in \cite{CalSom07,CalSom08}. However, in those works the posterior
density function $\pi(\B{x},\B{\theta}\mid \B{b})$ is maximized, yielding the
MAP estimator, whereas in this paper we want to perform uncertainty
quantification (UQ), and so need to compute samples from the posterior.
Sample-based methods for inverse problems first appear in the works
\cite{NicFox,KaiKolSomVau00}. In recent years, MCMC methods for Bayesian
inverse problems has become an active field, with some recent advances
including gradient and Hessian-based MCMC methods \cite{MarWilBurGha,PetraMartinStadlerEtAl14},
likelihood-informed MCMC methods \cite{CuiLawMar}, and transport map
accelerated MCMC methods \cite{ParMar}. In the Bayesian statistics literature,
MCMC methods for hierarchial models of the type considered here are standard;
see, e.g., \cite{GelCarSteRub04}. Moreover, the Gibbs sampler of \cite{Bar11}
for inverse problems is used in the context of spatio-temporal models in
\cite{Higdon}. Some properties of this Gibbs sampler are derived in
\cite{AgaBarPapStu}, and various extensions are presented in
\cite{brown2016computational,FoxNor,joyce2018point}, which have improved
convergence properties and/or improved computational efficiency. The algorithms
presented in this paper fit within this last group of MCMC methods.}

The paper is organized as follows. In~\cref{sec:review}, we present the
hierarchical Gibbs sampler and discuss the infinite dimensional limit,
presenting the result of \cite{AgaBarPapStu} showing the degeneracy of
hierarchical Gibbs as $N\rightarrow\infty$.  We then present the one-block
algorithm of \cite{RueHel}, which does not have the same degeneracy issues.
In~\cref{sec:alg}, we present the new algorithms making only limited
assumptions regarding the approximate full posterior distribution.
In~\cref{sec:analysis}, we provide a theoretical analysis of our proposed
algorithms, making explicit use of low rank structure.  The numerical
experiments in~\cref{sec:num} include a model 1D deblurring problem which
allows us to compare and contrast the various algorithms and a PDE-based
example of inverse heat equation that demonstrates the computational benefits
of our approaches. We summarize our work and discuss future research
in~\cref{sec:conclusions}. 

\section{Review of MCMC Algorithms}
\label{sec:review}
In this section, we present two MCMC algorithms for background. The first method is known as hierarchical Gibbs and is standard. Its convergence characteristics
serve as motivation for the second algorithm, which is known as one-block.

\subsection{The Hierarchical Gibbs Sampler}
Under our assumed model, we provide explicit expressions for the posterior and the marginal distribution. To obtain the expression for the posterior distribution, combine~\cref{BayesLaw} with~\cref{hyperprior} via~\cref{eqn:post} to obtain
\[\begin{aligned}
\pi(\B{x},\B{\theta}\mid\B{b})\propto&\> {\pi(\B{x}\mid\B{b},\B{\theta})\pi(\B{\theta})}\nonumber\\
\propto& \> \mu^{M/2}\sigma^{N/2}\pi(\B{\theta})\exp\left(-\frac\mu2\Vert\B{A}\B{x}-\B{b}\Vert^2-\frac\sigma2\B{x}^\top\prior^{-1}\B{x}\right),\nonumber\\
\propto & \>\mu^{M/2}\sigma^{N/2}\pi(\B{\theta}) \exp\left(-\frac{\mu}{2}\B{b}^\top \B{b} +\mu\B{b}^\top\B{Ax} -\frac{1}{2}\B{x}^\top\cond\B{x} \right) \nonumber\\
 =& \> \mu^{M/2}\sigma^{N/2}\pi(\B{\theta}) \exp\left(-\frac{\mu}{2}\B{b}^\top \B{b} + \frac{\mu^2}{2}\B{b}^\top\B{A}\cond\B{A}^\top \B{b} \right) \times \label{e_post2}\\
&  \qquad \qquad {\exp\left(- \frac{1}{2} (\B{x}-\xcond)^T\condinv(\B{x}-\xcond)\right), }\nonumber
\end{aligned}\]
where $\xcond = \mu\cond \B{A}^\top\B{b}$ and $\condinv = \mu\B{A}^\top\B{A} +
\sigma \prior^{-1}.$ It follows that the marginal distribution, $\pi(\B{\theta}\mid \B{b}) = \int_{\mathbb{R}^n}\pi(\B{x}, \B{\theta}\mid \B{b})d\B{x}$, is given by
\begin{equation}
\pi(\B{\theta} \mid \B{b}) \propto  \frac{\mu^{M/2}\sigma^{N/2}\pi(\B{\theta})}{\sqrt{\det(\condinv)}} \exp\left(-\frac{\mu}{2}\B{b}^\top \B{b} + \frac{\mu^2}{2}\B{b}^\top\B{A}\cond\B{A}^\top \B{b}\right). \label{marginal}
\end{equation}
Observe that the joint posterior density satisfies $\pi(\B{x},\B{\theta} \mid \B{b}) = \pi(\B{\theta}\mid  \B{b}) \pi(\B{x}\mid\B{\theta},\B{b})$, where $\B{x} \mid \B{\theta},\B{b} \sim \mathcal{N}(\xcond,\cond)$.

We begin with the hierarchical Gibbs sampler of \cite{Bar11}. Our choice of prior \cref{prior} for $\B{x}$, and the hyper-prior \cref{hyperprior} for $\B{\theta}=(\mu,\sigma)$, respectively, were made with conjugacy relationships in mind \cite{GelCarSteRub04}; i.e., so that the `full conditional' densities have the same form as the corresponding priors:
\begin{align}\label{conditional2}
\pi(\mu\mid\B{x},\sigma,\B{b})\propto& \> \mu^{M/2+\alpha_\mu-1}\exp\left(\left[-\frac{1}{2}\Vert\B{A}\B{x}-\B{b}\Vert^2_2-\beta_\mu\right]\mu\right),\\\label{conditional3}
\pi(\sigma\mid\B{x},\mu,\B{b})\propto& \> \sigma^{N/2+\alpha_\sigma-1}\exp\left(\left[-\frac{1}{2}\B{x}^T\prior^{-1}\B{x}-\beta_\sigma\right]\sigma\right),\\\label{conditional1}
\pi(\B{x}\mid\mu,\sigma,\B{b})\propto&\> \exp\left(-\frac{\mu}{2}\Vert\B{A}\B{x}-\B{b}\Vert^2-\frac{\sigma}{2}\B{x}^T\prior^{-1}\B{x}\right).
\end{align}
Note that~\cref{conditional2,conditional3} are Gamma densities, while~\cref{conditional1} is the density of a Gaussian distribution. \cref{ALG:HierarchicalGibbs} follows immediately from \cref{conditional2,conditional3,conditional1} and is precisely the hierarchical Gibbs sampling algorithm of \cite{Bar11}.

\LinesNumbered
\begin{algorithm}[!ht]
\DontPrintSemicolon
\SetKwInput{Input}{Input}
\SetKwInput{Output}{Output}
	\Input{Set 
$\B{x}_{(0)}=\B{x}_{\text{cond}}(\B{\theta}_{(0)})$, and define $K$ and burn-in period $K_b$.}
	\Output{Approximate samples from the posterior distribution  $\{ \B{x}_{(t)}, \B{\theta}_{(t)} \}_{t=K_b +1}^K$.}
	\BlankLine
  	\For{$t=1$ to $K$}
		{
Compute $\mu_{(t)}\sim \Gamma\left(M/2+\alpha_\mu,\frac{1}{2}\Vert\B{A}\B{x}_{(t-1)}-\B{b}\Vert^2+\beta_\mu\right)$.\;
Compute $\sigma_{(t)}\sim \Gamma\left(N/2+\alpha_\sigma,\frac{1}{2}(\B{x}_{(t-1)})^T\B{\Gamma}_{\rm prior}\B{x}_{(t-1)}+\beta_\sigma\right)$.\;
	Compute $\B{x}_{(t)}\sim \N\left( \B{x}_\text{cond} (\B{\theta}_{(t)}),\B{\Gamma}_{\rm cond}(\B{\theta}_{(t)})\right)$, where $\B{\theta}_{(t)}=(\mu_{(t)},\sigma_{(t)})$.
        }
\caption{Hierarchical Gibbs Sampler}
\label{ALG:HierarchicalGibbs}
\end{algorithm}



The values of $M$ and $N$ are determined by the number of measurements and the size of the numerical mesh, respectively, making the problems discrete. Since $M$ is the dimension of our measurement vector, we assume that it is a fixed value.  However, we are free to choose $N$ as we please, and the behavior of our approaches as $N\to\infty$ is an important question.  In what follows, we briefly discuss the infinite-dimensional limit, pointing the interested reader to the extensive treatments found in \cite{Stuart10,DashtiStuart16} for more details.

Consider the linear inverse problem, which typically arises from the discretization of a Fredholm integral equation of the first-kind, for example
\begin{equation}\label{eqn:fred}
b(s)=\int_{\Omega}a(s;t)x(t)dt,\quad s\in\Omega,
\end{equation}
where $b$ is the model output function, $\Omega$ is the computational domain, $a$ is the integral kernel or point spread function, and $x$ is the unknown which we seek to estimate. We define $\mathcal{A}_M x$ to be the forward operator discretized only in its range, so that $\mathcal{A}_M:X \to \mathbb{R}^M$, where $X = C(\bar\Omega)$, and $\Omega$ is the spatial domain. For example, in one-dimensional deconvolution, (with $\Omega = (0, 1)$) one can have
\[
   [\mathcal{A}_M x]_i= \int_{0}^{1} a(s_i-s') x(s') ds',\quad i=1,\ldots,M.
\]

Discretizing this integral in the $s'$ variable, e.g., using a uniform mesh on $[0,1]$ with $N$ grid elements and midpoint quadrature, then yields $\B{A}\B{x}$. Then we have that
\begin{equation}\label{equ:limmisfit}
\lim_{N\to\infty}\Vert\B{A}\B{x}-\B{b}\Vert^2_2=\Vert\mathcal{A}_M x-\B{b}\Vert^2_2.
\end{equation}
Note that here $\B{x}$ denotes the discretized version of $x \in X$. For the prior \cref{prior}, it is typical to choose $\B{\Gamma}_{\rm prior}$ to be the numerical discretization of the inverse of a differential operator. That is, letting $\Gamma_{\text{prior}}$ denote the infinite dimensional prior covariance operator, we define $\Gamma_\text{prior} = \mathcal{L}^{-1}$, where $\mathcal{L}$ is a differential operator. A basic requirement on $\Gamma_\text{prior}$ is that it has to be trace-class on $L^2(\Omega)$. That is, for any orthonormal basis $\{\phi_i\}_{i=1}^\infty$ of $L^2(\Omega)$, $\sum_{i=1}^\infty\langle \phi_i,\mathcal{L}^{-1}\phi_i\rangle<\infty$; see \cite{Stuart10}. Moreover, if $\langle
x,y\rangle$ is the standard $L^2(\Omega)$ inner product, we can use midpoint quadrature to obtain $\langle x,y\rangle=\lim_{N\to\infty} \langle \B{x}, \B{y} \rangle_N$, where the boldface letters indicate discretized versions of the variables and $\langle \B{x}, \B{y} \rangle_N := \frac1N\sum_{i=1}^N x_iy_i$.
Using this notation,
\begin{equation}\label{equ:limprior}
\lim_{N\to \infty}
   \ip{\B{x}}{\priorcov^{-1} \B{x}}_N = \ip{x}{\mathcal{L}x}.
\end{equation}
Combining~\cref{equ:limmisfit} and~\cref{equ:limprior} yields
\begin{equation}
\label{InfiniteDimemsions}
\lim_{N\to \infty}\left\{\frac\mu2\Vert\B{A}\B{x}-\B{b}\Vert^2_2
+\frac{\sigma}{2}\ip{\B{x}}{\priorcov^{-1}\B{x}}_N\right\}
=\frac\mu2\Vert\mathcal{A}_M x-\B{b}\Vert^2_2+\frac{\sigma}{2}\langle x,\mathcal{L}x\rangle.
\end{equation}

{There are two issues that arise when we consider the limit as $N\rightarrow \infty$: the first is mathematical, and the second is computational.}
A question that immediately arises is whether or not one can define an infinite dimensional limit of the posterior density function. We cannot define a probability density function on a function space. The reason for this is that in finite dimensions, the posterior density is none but the Radon--Nikodym derivative of posterior probability law of the inference parameter with respect to Lebesgue measure, but one cannot define a Lebesgue measure on an infinite-dimensional function space. However, as shown in~\cite{Stuart10}, given that the prior covariance operator is trace class and $\mathcal{A}_M:X\to\mathbb{R}^M$ is a continuous linear transformation, the posterior law of $x$, which we denote by
{
$\nu^{\B{b},\B{\theta}}$ is a Gaussian measure on $L^2(\Omega)$,
$\nu^{\B{b},\B{\theta}} = \mathcal{N}(x_\text{cond}(\B\theta), \Gamma_\text{cond}(\B\theta))$, with
\[
\Gamma_\text{cond}^{-1} = \mu \mathcal{A}_M^* \mathcal{A}_M + \sigma \Gamma_{\text{prior}}^{-1}, \qquad
x_\text{cond} =\mu\Gamma_\text{cond}\mathcal{A}^*_M \B{b}.
\]
}%
As for the construction of the prior covariance operator, as mentioned before, a common approach is to define them as inverses of differential operators. For example, we can define
{
\[
    \Gamma_\text{prior} = \mathcal{L}^{-s}, \quad
    \mathcal{L}u = -\kappa \Delta u + \alpha u,
    \quad \kappa > 0, \alpha \geq 0,
\]
with suitable boundary conditions, which is related to the Whittle-Mat\'ern
prior \cite{bardsley2018}. In this paper, we choose $\alpha=0$, $\kappa = \sigma^{-1/s}$, and  homogeneous Dirichlet boundary conditions for the Laplacian. Therefore,  $\Gamma_\text{prior} = \sigma (-\Delta)^{-s}$, leaving only $\sigma$ as the hyper-parameter for the prior.} Moreover, to
ensure that the covariance operator so defined is trace-class, we require $s >
d / 2$, where $d$ is the spatial dimension of the problem. For the
one-dimensional example in \cref{eqn:fred}, $s = 1$ would suffice. For problems
with $d = 2$ or $d =3$, a convenient option is $s = 2$. Note this assumption on
$s$ also ensures that the prior draws are almost surely continuous. For further
details on the definition of Gaussian measures on infinite-dimensional Hilbert
spaces, see~\cite{DaPrato06,DaPratoZabczyk02}.

A second question that arises is whether or not the performance of the hierarchical Gibbs sampler is dependent upon $N$. For inverse problems in which the infinite dimensional limit is well-defined, MCMC methods whose performance is independent of the discretization ($N$ in this case) are desirable. In line 3 of \cref{ALG:HierarchicalGibbs}, we see that $N$ appears in the Gamma conditional density $\pi(\sigma\mid\B{x},\mu,\B{b})$, thus it should not be surprising to find the $\sigma$-chain is dependent on $N$. The exact nature of this dependence is the subject of \cite[Theorem 3.4]{AgaBarPapStu}, where under reasonable assumptions it is shown that  the expected step in the $\sigma$-chain scales like $2/N$. Specifically, for any $\sigma>0$,
$$
\frac{N}{2}\mathbb{E}\left[\sigma_{(t+1)}-\sigma_{(t)}|\sigma_{(t)}=\sigma\right]=(\alpha_\sigma+1)\sigma
-f_N(\sigma;\bb)\sigma^2+\mathcal{O}(N^{-1/2}),
$$
where $\mathbb{E}$ denotes expectation and $f_N(\sigma;\bb)$ is bounded uniformly in $N$. Moreover, the variance of the step also scales like $2/N$; for any $\sigma>0$,
$$
\frac{N}{2}{\rm Var}\left[\sigma_{(t+1)}-\sigma_{(t)}|\sigma_{(t)}=\sigma\right]=2\sigma^2+\mathcal{O}(N^{-1/2}).
$$
A consequence of these results, as is noted in \cite{AgaBarPapStu}, is that the expected squared jumping distance of the Markov chain for $\sigma$ is $\mathcal{O}(1/N)$. Moreover, it is noted that the lag-1 autocorrelation of the $\sigma$-chain behaves like $1-c/N$ for some constant, but ${\rm Var}(\sigma_{(t)})=\mathcal{O}(1)$. Hence, the Monte Carlo error associated with $K-K_b$ draws in stationarity is $\mathcal{O}(\sqrt{N/(K-K_b)})$. {Thus,} the $\sigma$-chain becomes increasing correlated as $N\rightarrow\infty$. This phenomenon is illustrated in \cref{fig:ACF}, which displays the empirical autocorrelation functions for the $\mu$- and $\sigma$-chains generated by \cref{ALG:HierarchicalGibbs} for a one-dimensional image deblurring test problem. Note that as $N$ increases by a power of $2$, so does the integrated autocorrelation time (IACT).

\begin{figure}\centering
\includegraphics[scale=0.4]{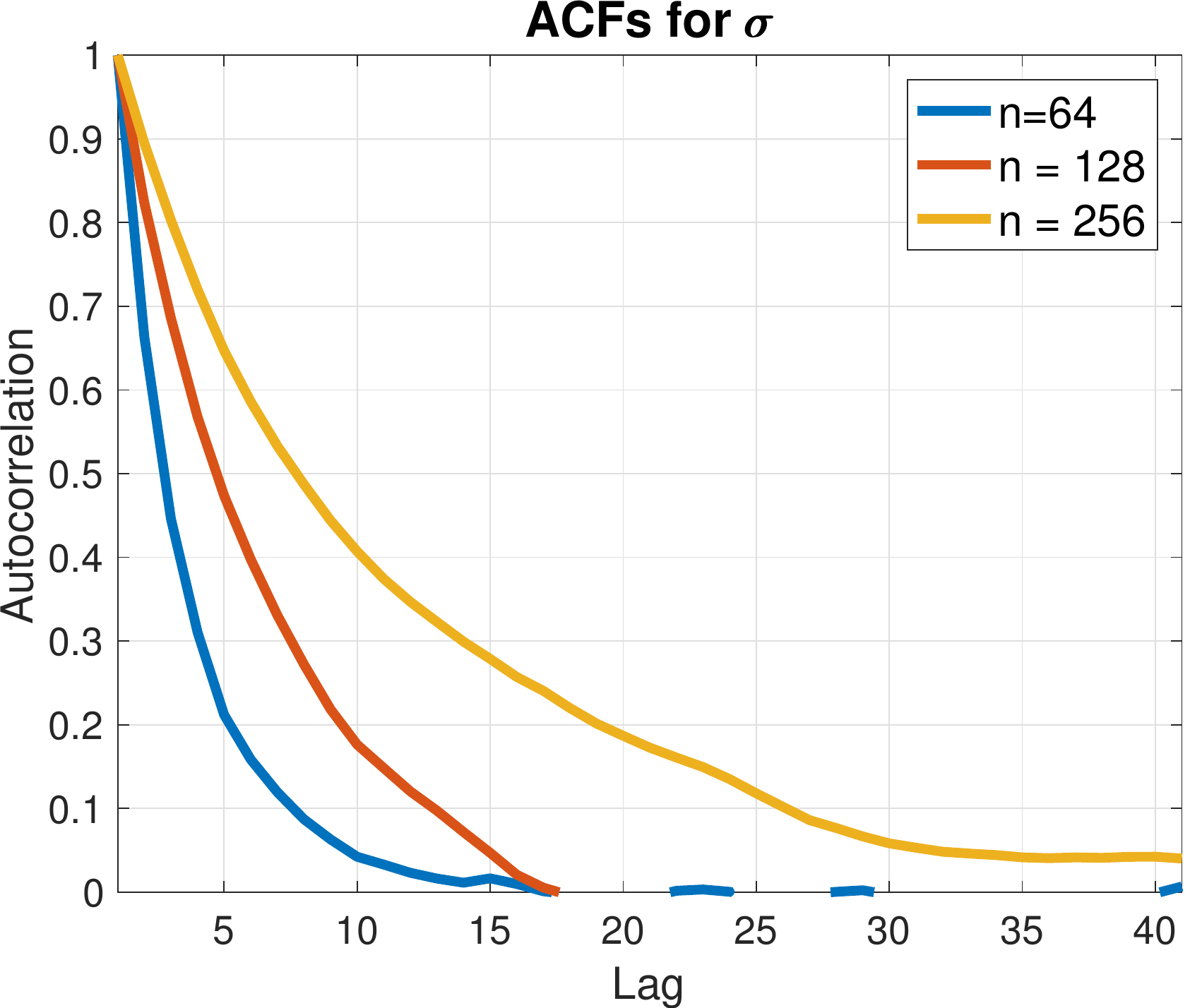} \hspace{1mm}
\includegraphics[scale=0.4]{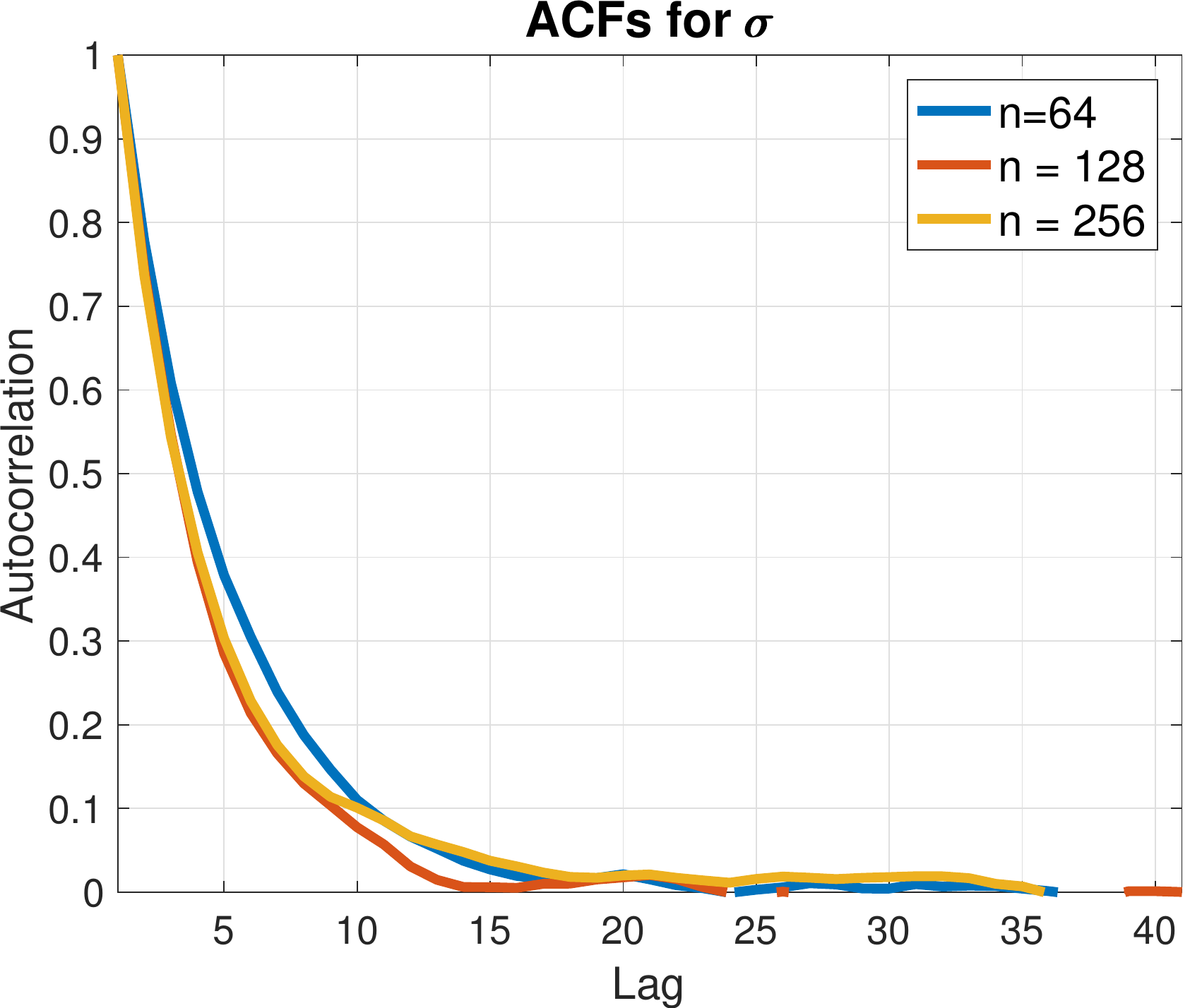}
\caption{Plots of the {empirical} autocorrelation functions for the $\sigma$ chains in the one-dimensional deblurring example: (left) Block-Gibbs and (right) Approximate One-Block discussed in~\cref{sec:OneBlockApprox}. See~\cref{ssec:1d} for details regarding the experimental setup. }
\label{fig:ACF}
\end{figure}

\subsection{One-Block MCMC}
\label{sec:OneBlock}

An alternative algorithm for computing a sample $(\B{x}',\B{\theta}')\sim \pi(\B{x}',\B{\theta}'\mid\B{b})$ is to first compute a sample from the marginal density $\B{\theta}'\sim \pi(\B{\theta}'\mid\B{b})$, defined in~\cref{marginal}, and then compute a sample from the conditional density $\B{x}'\sim \pi(\B{x}'\mid\B{b},\B{\theta}')$. In principle, we can use this procedure as an alternative to \cref{ALG:HierarchicalGibbs}. In practice, it is often not possible to sample directly from $\pi(\B{\theta}\mid\B{b})$. In the {\em one-block algorithm} of \cite{RueHel}, a Markov chain is generated in which if $(\B{x},\B\theta)$ is the current element of the chain, a state $\B{\theta}'$ is proposed from some proposal distribution with density $r(\B{\theta}'\mid\B{\theta})$, then $\B{x}'\sim\pi(\B{x}'\mid\B{b},\B{\theta}')$ is drawn, and the pair $(\B{x}',\B{\theta}')$ is accepted with probability
\begin{align}
    \rho_1(\B{\theta}';\B{\theta}) =& \> \min\left\{1,\frac{\pi(\B{x}',\B{\theta}'\mid\B{b})}{\pi(\B{x},\B{\theta}\mid\B{b})}
    \frac{r(\B{\theta}\mid \B{\theta}')\pi(\B{x}\mid\B{b},\B{\theta})}{r(\B{\theta}'\mid\B{\theta})\pi(\B{x}'\mid\B{b},\B{\theta}')}\right\}\nonumber \\ \label{OneBlockRatio}
    =&\> \min\left\{1,\frac{\pi(\B{\theta}'\mid\B{b})r(\B{\theta} \mid \B{\theta}')}{\pi({\B{\theta}}\mid\B{b})r(\B{\theta}' \mid \B{\theta})}\right\}.
\end{align}
The resulting MCMC method is given by \cref{ALG:OneBlock}.
\LinesNumbered
\begin{algorithm}[!ht]
\DontPrintSemicolon
\SetKwInput{Input}{Input}
\SetKwInput{Output}{Output}
	\Input{Set 
$\B{x}_{(0)}=\B{x}_\text{cond}(\B{\theta}_{(0)})$, and define $K$ and burn-in period $K_b$.}
	\Output{Approximate samples from the posterior distribution  $\{ \B{x}_{(t)}, \B{\theta}_{(t)} \}_{t=K_b +1}^K$.}
	\BlankLine
  	\For{$t=1$ to $K$}
		{
   	    Compute $\B{\theta}_{\ast}\sim r(\B{\theta}_*\mid \B{\theta}_{(t-1)})$. \;
        Compute $\B{x}_{(t)}\sim \pi(\B{x}_{(t)}\mid \B{b},\B{\theta}_{\ast})$.\;
        Set $(\B{x}_{(t)},\B{\theta}_{(t)})=(\B{x}_\ast,\B{\theta}_{\ast})$ with probability
        $\rho_1(\B{\theta}_{\ast};\B{\theta}_{(t-1)})$, defined in~\cref{OneBlockRatio},
        else set $(\B{x}_{(t)},\B{\theta}_{(t)})=(\B{x}_{(t-1)},\B{\theta}_{(t-1)})$. \;
        }
\caption{One-Block MCMC}
\label{ALG:OneBlock}
\end{algorithm}

{

Note that by combining the random walk $\B{\theta}'\sim r(\B{\theta}'\mid \B{\theta})$ with the conditional sample $\B{x}'\sim\pi(\B{x}'\mid\B{b},\B\theta')$,  $\B{x}$ is marginalized (i.e., integrated out) from the posterior, as is seen by the acceptance ratio \cref{OneBlockRatio} which is independent of $\B{x}$. Note that if $r(\cdot \mid \cdot)$ is used as a Metropolis-Hastings proposal for sampling from the marginal density $\pi(\B{\theta}\mid\B{b})$ (which defines the MCMC method proposed in \cite{FoxNor}), the acceptance ratio is also given by \cref{OneBlockRatio}. Thus, the $\B\theta$-chain generated by \cref{ALG:OneBlock} converges in distribution to $\pi(\B{\theta}\mid\B{b})$ and its behavior is independent of $\B{x}$, as desired.

As an alternative to MCMC, we remark that there exist special cases of the hierarchical Bayesian model in which the hyperparmeters $\B{\theta}$ can be analytically integrated out to obtain a closed-form marginal posterior density $\B{x} \mid \B{b}$ (e.g., multivariate $t$). Despite the fact that MCMC is no longer necessary in such scenarios, there is still the need to work with a scale matrix that may require many forward solutions to compute, rendering such an approach impractical. \cite{WilcoxGhattas} study this issue in the Gaussian case in which the precision parameters are held fixed (as opposed to integrating them out). The same problems arise with any other distribution of $\B{x} \mid \B{b}$, but with additional difficulties since other forms are generally no longer in the Gaussian class of distributions. Further, a researcher may wish to use alternative hyperpriors other than those that lead to closed-form distributions for the sake of obtaining a better reconstruction. The approach we consider here is flexible with respect to such considerations because the conditional distribution of $\B{x} \mid \B{\theta}, \B{b}$ is still Gaussian, regardless of the prior on $\B{\theta}$. Exploration of alternative priors for $\B{\theta}$ is beyond the scope of the present work.
}

\section{Algorithms}
\label{sec:alg}

{For some applications,} evaluating the marginal density $\pi(\B{\theta}\mid\B{b})$ is not computationally tractable, nor is computing exact samples $\B{x}'\sim \pi(\B{x}'\mid \B{b},\B{\theta}')$. In this case, implementation of~\cref{ALG:OneBlock} is infeasible. However, suppose that we have an approximate posterior density function $\hat\pi(\B{x},\B{\theta}\mid\B{b})$ for which the marginal density $\hat\pi(\B{\theta}\mid\B{b})$ can be evaluated, and exact samples $\B{x}'\sim \hat\pi(\B{x}'\mid \B{b},\B{\theta}')$ can be drawn efficiently.

In this section, we propose three MCMC algorithms, each of which uses $\hat\pi(\B{x},\B{\theta}\mid\B{b})$ to approximate~\cref{ALG:OneBlock}. We note that Rue and Held \cite{RueHel} also consider approximate marginalization of $\B{x}$ in the one-block algorithm. {However, the difference between our approach and that of \cite{RueHel} is the nature of the approximate posterior distribution. Specifically, motivated by non-Gaussian conditional distributions that arise in, e.g., Poisson counts for disease mapping, \cite{RueHel} use an approximation  based on a Taylor expansion of the log-likelihood about the mode of the full conditional distribution, resulting in a Gaussian proposal in the Metropolis-Hastings algorithm. By contrast, we are concerned with cases in which the full conditional is still Gaussian distributed, but drawing realizations from the distribution and evaluating the density are computationally prohibitive because of the dimensionality of the unknown parameters.} We discuss in \cref{ss_lowrank} our construction of the approximation $\hat\pi(\B{x},\B{\theta} \mid \B{b})$ and its effect on the proposed algorithms.

\subsection{Approximate One-Block MCMC}\label{sec:OneBlockApprox}
Suppose we generate a proposal $(\B{x}',\B\theta')$ by first drawing $\B\theta'$ from a proposal distribution with density $r(\B\theta'\mid\B\theta)$ followed by drawing $\B{x}'\sim \hat\pi(\B{x}'\mid\B{\theta}',\B{b})$. Using this in a Metropolis-Hastings sampler with target distribution $\pi(\B{x}, \B\theta \mid \B{b})$ is simply the one-block algorithm with $\hat\pi(\B{x} \mid \B\theta, \B{b})$ in place of $\pi(\B{x} \mid \B{\theta}, \B{b})$. The proposal density is then given by
\[
    q(\B{x}',\B\theta'\mid\B{x},\B\theta) = \hat\pi(\B{x}'\mid \B\theta',\B{b})r(\B\theta'\mid\B\theta),
\]
and hence, the acceptance probability is
\begin{equation}
\label{OneBlockApproxRatio}
\rho_2(\B{x}',\B{\theta}'; \B{x},\B{\theta})=\min\left\{1,\frac{\pi(\B{x}',\B{\theta}' \mid \B{b})\hat\pi(\B{x}\mid\B{\theta},\B{b})r( \B{\theta}\mid\B{\theta}' )}{\pi(\B{x},\B{\theta} \mid \B{b})\hat\pi(\B{x}'\mid\B{\theta}',\B{b})r(\B{\theta}'\mid \B{\theta})} \right\}.
\end{equation}
The full procedure is summarized in~\cref{ALG:OneBlockApprox}.
\LinesNumbered
\begin{algorithm}[tb]
\DontPrintSemicolon
\SetKwInput{Input}{Input}
\SetKwInput{Output}{Output}
	\Input{Initialize $\B{\theta}_{(0)}$. Define the length of the chain $C$ and burn-in period $C_b$.}
	\Output{Approximate sample from the posterior distribution  $\{ \B{x}_{(t)}, \B{\theta}_{(t)} \}_{t=C_b +1}^K$.}
	\BlankLine
	Compute $\B{x}_{(0)}\sim \hat\pi(\B{x}\mid\B{b},\B{\theta}_{(0)})$. \;
  	\For{$t=1$ to $K$} {
		Draw $\B{\theta}'$ from distribution with density $r(\B{\theta}'\mid \B{\theta}_{(t-1)})$. \;
		Draw $\B{x}'$ from distribution with density $\hat\pi(\B{x}'\mid\B{\theta}',\B{b})$\;
		Set $(\B{x}_{(t)},\B\theta_{(t)} ) = (\B{x}',\B{\theta}')$ with probability $\rho_2(\B{x}',\B\theta'; \B{x}_{(t-1)},\B{\theta}_{(t-1)})$ defined in \cref{OneBlockApproxRatio}; else set $(\B{x}_{(t)},\B\theta_{(t)}) =  (\B{x}_{(t-1)},\B{\theta}_{(t-1)})$.  \;
        }
\caption{Approximate One-Block MCMC.}
\label{ALG:OneBlockApprox}
\end{algorithm}

There are two observations to make about~\cref{ALG:OneBlockApprox}, each of which motivates the MCMC methods that follow. First, computing the ratio~\cref{OneBlockApproxRatio} requires evaluating the true posterior density $\pi(\B{x},\B{\theta} \mid \B{b})$. In cases in which this is computationally burdensome, the MCMC method presented next seeks to avoid extraneous evaluations of $\pi(\B{x},\B{\theta} \mid \B{b})$ using a technique known as {\em delayed acceptance} \cite{ChrFox}. Second,~\cref{ALG:OneBlockApprox} only {\em approximately} integrates $\B{x}$ out of $\pi(\B{x}, \B\theta \mid \B{b})$, since $\pi(\B\theta \mid \B{b}) \approx \pi(\B{x}, \B{\theta} \mid \B{b})/\hat\pi(\B{x} \mid \B\theta, \B{b})$, but this is not an exact marginalization and thus some dependence between the $\B{x}$ and $\B{\theta}$ chains {may} remain, slowing convergence of the algorithm. However, as we discuss below,~\cref{ALG:OneBlockApprox} is a special case of {the} so-called  {\em pseudo-marginal MCMC algorithm} \cite{PseudoMarg}. By generalizing the approximate one-block algorithm, we can improve the approximation to $\pi(\B{\theta} \mid \B{b})$ and thus improve {the} mixing of the Markov chains.

\subsection{Approximate One-Block MCMC with Delayed Acceptance}\label{sec:OneBlockDA}
It will often be the case that evaluating $\pi(\B{x},\B{\theta}\mid\B{b})$ is computationally prohibitive and/or significantly more expensive than evaluating $\hat\pi(\B{x},\B{\theta}\mid\B{b})$, in which case one would like to generate samples from $\pi(\B{x}, \B{\theta} \mid \B{b})$ while minimizing the number of times its density is evaluated. This motivates the use of the {\em delayed acceptance} framework of \cite{ChrFox} to improve the computational efficiency of~\cref{ALG:OneBlockApprox}. In this approach, one step of~\cref{ALG:OneBlock} is used with $\hat\pi(\B{x},\B{\theta} \mid \B{b})$ taken as the target distribution. Only if $(\B{x}',\B{\theta}')$ is accepted as a sample from $\hat\pi(\B{x},\B{\theta}\mid\B{b})$ is it proposed as a sample from $\pi(\B{x},\B{\theta}\mid\B{b})$. The idea is to only evaluate $\pi(\B{x}', \B{\theta}' \mid \B{b})$ for proposed states that have a high probability of being accepted as draws from the true distribution. This prevents us from wasting computational effort on rejected proposals. The approximate distribution in the first stage is essentially a computationally cheap ``filter'' that prevents this from occurring. The resulting {\em approximate one-block MCMC with delayed acceptance (ABDA)} procedure is given in~\cref{ALG:OneBlockApproxDA}. {Note that the delayed acceptance algorithm proposed here is related to the surrogate transition method~\cite[Section 9.4.3]{Liu2004} and is a special case of \cite{ChrFox} with a state-independent approximation.}

\LinesNumbered
\begin{algorithm}[tb]
\DontPrintSemicolon
\SetKwInput{Input}{Input}
\SetKwInput{Output}{Output}
	\Input{Initialize $\B{\theta}_{(0)}$. Define the length of the chain {$C$ and burn-in period $C_b$}.}
	\Output{Approximate sample from the posterior distribution  $\{ \B{x}_{(t)}, \B{\theta}_{(t)} \}_{t=C_b +1}^C$.}
	\BlankLine
	Compute $\B{x}_{(0)}\sim \hat\pi(\B{x}\mid\B{b},\B{\theta}_{(0)})$. \;
 	\For{$t=1$ to $K$} {
		\tcp{\underline{Stage 1:} apply the one-block algorithm to $\hat\pi(\B{x},\B{\theta}\mid\B{b})$.}
		Draw $\B{\theta}_\ast$ from the distribution with density $r(\B{\theta}_\ast\mid\B{\theta}_{(t-1)})$. \;
		Draw $\B{x}_\ast$ from the distribution with density $\hat\pi(\B{x}_\ast\mid\B{\theta}_\ast,\B{b})$.\;
        Compute $\hat\rho_1(\B{\theta}_\ast; ~\B{\theta}_{(t-1)})$ defined in \cref{eqn:DAOneBlockRatio}\;
		Set $(\B{x}',\B{\theta}') \leftarrow (\B{x}_\ast, \B{\theta}_\ast)$ and promote to Stage 2 with probability $\hat\rho_1(\B{\theta}';\B{\theta}_{(t-1)})$ else set $(\B{x}',\B{\theta}') \leftarrow (\B{x}_{(t-1)},\B{\theta}_{(t-1)})$ and promote. \;
		\tcp{\underline{Stage 2:} accept/reject $(\B{x}',\B{\theta}')$ as a proposal for $\pi(\B{x},\B{\theta}\mid\B{b})$.}
        Define
	\begin{equation}
        \label{eqn:DAOneBlockRatio2}
        \rho_3(\B{x}', \B{\theta}';\B{x},\B{\theta}) = \min\left\{1,\frac{\pi(\B{x}',\B{\theta}' \mid \B{b})q(\B{x},\B{\theta}\mid \B{x}',\B{\theta}')}{\pi(\B{x},\B{\theta} \mid \B{b})q(\B{x}',\B{\theta}'\mid \B{x},\B{\theta}) } \right\},
        \end{equation}
	where $q(\B{x}',\B{\theta}'\mid \B{x}_{(t-1)},\B{\theta}_{(t-1)})$ is defined in~\cref{e_efftrans}.
        Set $(\B{x}_{(t)},\B{\theta}_{(t)})=(\B{x}',\B{\theta}')$ with probability $\rho_3(\B{x}', \B{\theta}';\B{x}_{(t-1)},\B{\theta}_{(t-1)})$, else set $(\B{x}_{(t)},\B{\theta}_{(t)})=(\B{x}_{(t-1)},\B{\theta}_{(t-1)})$.\;
		
	}
\caption{Approximate One-Block MCMC with delayed acceptance. }
\label{ALG:OneBlockApproxDA}
\end{algorithm}

Let $\hat\rho_1(\B{\theta}';~\B{\theta}_{(t-1)})$ be given by~\cref{OneBlockRatio}, but with {$\pi(\B{x},\B{\theta}\mid\B{b})$ replaced by $\hat\pi(\B{x},\B{\theta}\mid\B{b})$;} i.e.,
\begin{align}\label{eqn:DAOneBlockRatio}
    \hat\rho_1(\B{\theta}'; ~\B{\theta}_{(t-1)}) =& \min\left\{1,\frac{\hat\pi(\B{x}',\B{\theta}'\mid\B{b})}{\hat\pi(\B{x}_{(t-1)},\B{\theta}_{(t-1)}\mid\B{b})}
    \frac{r(\B{\theta}_{(t-1)}\mid \B{\theta}')\hat\pi(\B{x}_{(t-1)}\mid\B{b},\B{\theta}_{(t-1)})}{r(\B{\theta}'\mid\B{\theta}_{(t-1)})\hat\pi(\B{x}'\mid\B{b},\B{\theta}')}\right\}
    \nonumber \\
    =& \min\left\{1,\frac{\hat\pi(\B{\theta}'\mid\B{b})r(\B{\theta}_{(t-1)} \mid \B{\theta}')}{\hat\pi({\B{\theta}}_{(t-1)}\mid\B{b})r(\B{\theta}' \mid \B{\theta}_{(t-1)})}\right\}.
\end{align}
Then, Stage 2 of~\cref{ALG:OneBlockApproxDA} is a Metropolis-Hastings algorithm with target distribution $\pi(\B{x}, \B{\theta} \mid \B{b})$ and proposal density given by
\begin{equation}\label{e_efftrans}
\begin{aligned}
    q(\B{x}',\B\theta' \mid \B{x}_{(t-1)}, \B\theta_{(t-1)}) =&\> \hat\rho_1(\B{\theta}'; ~\B{\theta}_{(t-1)})\hat\pi(\B{x}' \mid \B{\theta}', \B{b})r(\B{\theta}' \mid \B\theta_{(t-1)}) \\
        &\quad + ~\delta_{(\B{x}_{(t-1)}, \B{\theta}_{(t-1)})}(\B{x}', \B{\theta}')(1 - y(\B{\theta}_{(t-1)})),
\end{aligned}
\end{equation}
where $y(\B{\theta}_{(t-1)}) = \int \int \hat\rho_1(\B{\theta}'; ~\B{\theta}_{(t-1)})\hat\pi(\B{x}' \mid \B{\theta}', \B{b})r(\B{\theta}' \mid \B{\theta}_{(t-1)})\dx \B{x}' \dx \B{\theta}'$. Note that there is never a need to evaluate $y(\B{\theta}_{(t-1)})$, since when the promoted state is $(\B{x}_{(t-1)}, \B{\theta}_{(t-1)})$, $$\rho_3(\B{x}',\B{\theta}';\B{x}_{(t-1)},\B{\theta}_{(t-1)})=1,$$ and the chain remains at the same point. Conversely, when the composite sample $(\B{x}', \B{\theta}')$ {is promoted, $(\B{x}', \B{\theta}') \neq (\B{x}_{(t-1)}, \B{\theta}_{(t-1)})$ and}
\[
    q(\B{x}',\B\theta' \mid \B{x}_{(t-1)}, \B\theta_{(t-1)}) = \hat\rho_1(\B{\theta}'; ~\B{\theta}_{(t-1)})\hat\pi(\B{x}' \mid \B{\theta}', \B{b})r(\B{\theta}' \mid \B\theta_{(t-1)}).
\]
We discuss in~\cref{sec:analysis} the conditions under which the acceptance rate in the second stage is high, thus preventing wasted computational effort on rejected proposals.

{The validity of the ABDA algorithm can be gleaned by recognizing that it is a special case of the delayed acceptance algorithm in \cite{ChrFox}, with first stage proposal given by $\hat\pi(\B{x} \mid \B{\theta}', \B{b})r(\B{\theta}' \mid \B{\theta})$, the approximating density given by $\hat\pi(\B{x}, \B{\theta} \mid \B{b})$, and target density $\pi(\B{x}, \B{\theta} \mid \B{y})$. This latter observation allows us to apply Theorem 1 of \cite{ChrFox} to establish irreducibility and aperiodicity of the ABDA Markov chain. Standard ergodic theory (e.g., \cite{RobertCasella04}) then provides that $\pi(\B{x}, \B{\theta} \mid \B{b})$ is, in fact, the limiting distribution.}

\subsection{Pseudo-Marginal MCMC}
As noted at the end of~\cref{ALG:OneBlock}, the $\B\theta$-chain generated by~\cref{ALG:OneBlock} is independent of the $\B{x}$-chain. Hence, it does not suffer from the same degeneracy issues as the $\B\theta$-chain generated by~\cref{ALG:HierarchicalGibbs}. When it is not computationally feasible to implement~\cref{ALG:OneBlock}, and an approximate posterior $\hat\pi(\B{x},\B\theta\mid\B{b})$ is used as in~\cref{ALG:OneBlockApprox,ALG:OneBlockApproxDA}, the marginalization is only approximate so that there still remains some dependence between the $\B{\theta}$ and $\B{x}$ chains. This dependence can be mitigated as the approximation to the marginal distribution of $\B{\theta}$ improves. In particular, we can use importance sampling, with importance density $\hat\pi(\B{x}\mid\B{b},\B\theta)$, to approximate the integration over $\B{x}$. Specifically, we have that
\begin{align}
	\pi(\B{\theta}'\mid\B{b})  =& \> \int_{\mb{R}^N} {\pi(\B{x},\B\theta'\mid\B{b})} d\B{x} \nonumber\\
                                              =& \> \int_{\mb{R}^N} \frac{\pi(\B{x},\B\theta'\mid\B{b})}{\hat\pi(\B{x}\mid\B{b},\B\theta')} \hat\pi(\B{x}\mid\B{b}, \B\theta')d\B{x}\nonumber \\ \label{pK}
	\approx& \>\frac{1}{K}\sum_{j=1}^{K} \frac{\pi(\B{x}'_j,\B\theta'\mid\B{b})}{\hat\pi(\B{x}'_j\mid\B{b},\B\theta')} \eqdef \pi^K (\B{\theta}'\mid\B{b}),
\end{align}
where $\B{x}'_j \sim \hat\pi(\B{x}'\mid \B{b},\B\theta')$. The idea behind pseudo-marginal MCMC \cite{PseudoMarg}, in our setting, is to generalize~\cref{ALG:OneBlockApprox} by using $\pi^K(\B{\theta}\mid\B{b})$ as an approximation to $\pi(\B{\theta}\mid \B{b})$. The resulting algorithm is given in~\cref{ALG:PseudoMarginal}.
\LinesNumbered
\begin{algorithm}[!ht]
\DontPrintSemicolon
\SetKwInput{Input}{Input}
\SetKwInput{Output}{Output}
	\Input{Initialize $\B{\theta}_{(0)}$. Define the length of the chain $C$ and burn-in period $C_b$.}
	\Output{Approximate samples from the posterior distribution  $\{ \{\B{x}_{(t,i)}\}_{i=1}^K, \B{\theta}_{(t)} \}_{t=C_b +1}^C$.}
	\BlankLine
    Compute  $\pi^K(\B{\theta}_{(0)}\mid\B{b})$ with $\B{x}_{j}'\sim\hat\pi(\B{x}'\mid \B{b},\B\theta_{(0)})$ as in \cref{pK}. Define $\{\B{x}_{(0,j)}\}_{j=1}^K=\{\B{x}'_j\}_{j=1}^K$.\; 	
    \For{$t=1$ to $C$} {
    Compute $\B{\theta}'\sim r(\cdot\mid\B{\theta}_{(t-1)})$.\;
	Compute $\pi^K(\B{\theta}'\mid\B{b})$ using $\B{x}'_j\sim\hat\pi(\B{x}'\mid \B{b},\B\theta')$ as in \cref{pK}.\;
    Compute
    \begin{equation}
    \label{PseudoMargRatio}
    \rho_4(\B{\theta}',\B{\theta}_{(t-1)})=\min\left\{1,\frac{\pi^K(\B{\theta}'\mid\B{b})r(\B{\theta}_{(t-1)}\mid\B{\theta}')}{\pi^K(\B{\theta}_{(t-1)}\mid\B{b})r(\B{\theta}'\mid\B{\theta}_{(t-1)})}\right\}.
    \end{equation}\;
	With probability $\rho_4(\B{\theta}',\B{\theta}_{(t-1)})$, set  $\B{\theta}_{(t)}=\B{\theta}'$ and $\{\B{x}_{(t,j)}\}_{j=1}^K=\{\B{x}'_j\}_{j=1}^K$,
else set $\B{\theta}_{(t)}=\B{\theta}_{(t-1)}$ and $\{\B{x}_{(t,j)}\}_{j=1}^K=\{\B{x}_{(t-1,j)}\}_{j=1}^K$.
}
\caption{Pseudo-Marginal MCMC. }
\label{ALG:PseudoMarginal}
\end{algorithm}

When $K=1$,~\cref{ALG:PseudoMarginal} simply reduces to~\cref{ALG:OneBlockApprox}. Conversely, as $K\rightarrow\infty$, $\pi^K (\B{\theta}'\mid\B{b})\rightarrow\pi(\B{\theta}'\mid\B{b})$, and hence the Markov chains produced by~\cref{ALG:PseudoMarginal,ALG:OneBlock} become indistinguishable. Consequently, as $K$ increases, the dependence between the $\B\theta$-chain and the $\B{x}$-chain dissipates as desired. It is shown in \cite{PseudoMarg} that the value of $\pi^K(\B{\theta}_{(t-1)}\mid\B{b})$ computed in step $t-1$ can be reused in step $t$ so that a new set of importance samples does not need to be computed in~\cref{pK}. Furthermore, this paper showed that the corresponding Markov chain will converge in distribution to the target density, in our case $\pi(\B{x},\B\theta\mid\B{b})$.

\section{Analysis}
\label{sec:analysis}

\cref{ALG:OneBlockApprox,ALG:OneBlockApproxDA,ALG:PseudoMarginal} are all approximations of the one-block algorithm, \cref{ALG:OneBlock}, and they are meant to be used in cases in which implementing one-block is computationally {expensive}. All three algorithms require an approximation of the posterior, $\hat\pi(\B{x},\B\theta\mid\B{b})$, which we assume has the form \cref{e_post2}, but with the conditional covariance $\cond$ replaced by an approximation $\condh$. We will use the following acronyms in what follows: AOB for Approximate One-Block, \cref{ALG:OneBlockApprox}, and ABDA for Approximate One-Block with Delayed Acceptance, \cref{ALG:OneBlockApproxDA}.

All three algorithms require the computation of samples from the approximate conditional $\hat\pi(\B{x}\mid\B{b},\B\theta)$, which is of the form $$\B{x}\mid\B{b},\B\theta \sim \mc{N}(\xcondh,\condh),$$ where
$\xcondh \equiv \mu \condh \B{A}^\top \B{b}$.
In~\cref{ss_lowrank}, we tailor these results to a specific choice of $\hat\pi(\B{x}\mid\B{b},\B\theta)$.
The following quantity will be important in what follows:
\begin{equation}\label{eqn:wx}
	w(\B{x},\B\theta) \equiv  \exp\left( -\frac{1}{2} \B{x}^\top(\condinv - \condhinv)\B{x} \right).
\end{equation}
An expression for the moments of $w(\B{x},\B\theta)$ can be computed analytically by using properties of Gaussian integrals and was established in~\cite{brown2016computational}. For positive integers $m$,
\begin{equation}\label{eqn:moments}
	\mathbb{E}_{\hat\pi(\B{x}\mid\B{b},\B\theta)} [w^m(\B{x},\B\theta)]  = \frac{1}{M_m(\B\theta)},
\end{equation}
where
\begin{equation}\label{eqn:nt}
	M_m(\B\theta) \equiv \frac{\exp\left(\frac{\mu^2}{2}\B{b}^\top\B{A}(\B{M}_m^{-1}(\B\theta)- \condh)\B{A}^\top \B{b} \right)}{(\det\condh {\det\B{\Gamma}_m(\B\theta))^{1/2} }},
\end{equation}
with
\[ {\B{\Gamma}_m(\B\theta)} = m(\condinv -\condhinv) + \condhinv. \]
Further results for $M_m(\B\theta)$ can be derived if $\condh$ is known explicitly. When $\condh$ is constructed using the low-rank approach outlined in \cref{ss_lowrank} below, we have that $M_m(\B\theta) \geq 1.$

\subsection{Analysis of the Approximate One Block acceptance ratio}


Our first result derives a simplified version of the acceptance ratio of AOB (\cref{ALG:OneBlockApprox}) that is more amenable to interpretation.
\begin{prop}\label{p_accept1}
In~\cref{ALG:OneBlockApprox}, let $(\B{x},\B\theta)$ denote the current state of the AOB chain and let $(\B{x}',\B\theta')$ be the proposed state. Then, the acceptance ratio simplifies to
	\[ \rho_2(\B{x}',\B{\theta}'; \B{x},\B{\theta}) =  \min\left\{1, \frac{z(\B{x}',\B{\theta}') \pi(\B{\theta}'\mid \B{b})r( \B{\theta}\mid\B{\theta}' )}{z(\B{x},\B{\theta})\pi(\B{\theta}\mid \B{b})r(\B{\theta}'\mid \B{\theta})}\right\},\]
	where $z(\B{x},\B\theta) \equiv w(\B{x},\B\theta)/M_1(\B\theta)$.
\end{prop}
{
\begin{proof}
It is clear that we only need to focus on the second term in the acceptance ratio, which we can rewrite as
\[ \frac{\pi(\B{x}',\B{\theta}' \mid \B{b})\hat\pi(\B{x}\mid\B{\theta},\B{b})r( \B{\theta}\mid\B{\theta}' )}{\pi(\B{x},\B{\theta} \mid \B{b})\hat\pi(\B{x}'\mid\B{\theta}',\B{b})r(\B{\theta}'\mid \B{\theta})} = \frac{ r( \B{\theta}\mid\B{\theta}' ) \pi(\B{x}',\B{\theta}' \mid \B{b})/ \hat\pi(\B{x}'\mid\B{\theta}',\B{b}) }{ r(\B{\theta}'\mid \B{\theta}) \pi(\B{x},\B{\theta} \mid \B{b})/\hat\pi(\B{x}\mid\B{\theta},\B{b}) } .\]
Since the posterior distribution is the product of the conditional and the marginal, we have
\begin{equation}\label{e_inter1}
\frac{\pi(\B{x},\B{\theta} \mid \B{b})}{\hat\pi(\B{x}\mid\B{\theta},\B{b})}= \frac{\pi(\B{x}\mid\B\theta, \B{b})\pi(\B\theta\mid\B{b})}{\hat\pi(\B{x}\mid\B{\theta},\B{b})}.
\end{equation}
In the proof of \cite[Proposition 1]{brown2016computational}, it is shown that $\pi(\B{x}\mid\B{\theta},\B{b})/\hat\pi(\B{x}\mid\B{\theta},  \B{b}) = z(\B{x},\B\theta)$, so that we get
\[ \frac{ r( \B{\theta}\mid\B{\theta}' ) \pi(\B{x}',\B{\theta}' \mid \B{b})/ \hat\pi(\B{x}'\mid\B{\theta}',\B{b}) }{ r(\B{\theta}'\mid \B{\theta}) \pi(\B{x},\B{\theta} \mid \B{b})/\hat\pi(\B{x}\mid\B{\theta},\B{b}) } = \frac{r( \B{\theta}\mid\B{\theta}' ) z(\B{x}',\B\theta')\pi(\B\theta'|\B{b})}{r(\B{\theta}'\mid \B{\theta}) z(\B{x},\B\theta)\pi(\B\theta|\B{b})}, \]
which gives the desired result.
\end{proof}
}

The meaning of this result is clear if we combine~\cref{eqn:moments} and~\cite[Proposition 2]{brown2016computational} to obtain
\begin{equation}\label{eqn:expw}
	\mathbb{E}_{\hat\pi(\B{x}\mid\B{b},\B\theta)}\,[z(\B{x},\B\theta)] = M_1^{-1}(\B\theta) \mathbb{E}_{\hat\pi(\B{x}\mid\B{b},\B\theta)}\,[w(\B{x},\B\theta)] = 1.
\end{equation}
{In other words, given the current state $(\B{x}, \B\theta)$ and a proposed state $\B{\theta}^\prime$, the ratio in \cref{ALG:OneBlockApprox}, when averaged over $\hat{\pi}(\B{x}^\prime \mid \B{\theta}^\prime, \B{b})$, is approximately that of the one-block algorithm, provided $z(\B{x}, \B\theta) \approx 1$.} Observe that if one takes $\condh = \cond$, so that $z(\B{x},\B\theta) = 1$, then
\[
\rho_2(\B{x}',\B{\theta}'; \B{x},\B{\theta}) =  \min\left\{1, \frac{ \pi(\B{\theta}'\mid \B{b})r( \B{\theta}\mid\B{\theta}' )}{\pi(\B{\theta}\mid \B{b})r(\B{\theta}'\mid \B{\theta})}\right\},
\]
which is exactly the one-block acceptance ratio \cref{OneBlockRatio}. Similarly, if $\condh \approx \cond$ is a sufficiently accurate approximation, then $z(\B{x},\B\theta) \approx 1$, and AOB will closely approximate the One-Block algorithm.
\subsection{Analysis of the ABDA acceptance ratios}

Next, we discuss the acceptance ratio of ABDA (\cref{ALG:OneBlockApproxDA}). In the first stage, one-block is applied to $\hat\pi(\B{x},\B{\theta} \mid \B{b})$, yielding the acceptance ratio \cref{eqn:DAOneBlockRatio} and target distribution $\hat\pi(\B\theta\mid \B{b})$. The analysis for the acceptance rate at the second stage {is provided in \cref{p_twostage}.}
\begin{prop}\label{p_twostage}
Let $(\B{x},\B\theta)$ denote the current state of the ABDA chain and let $(\B{x}',\B\theta')$ be the proposed state. 
{With  $w(\B{x},\B{\theta})$ defined in~\cref{eqn:wx}, the acceptance ratio at the second stage is}
	\[ \rho_3(\B{x}',\B\theta';\B{x},\B\theta) = \min\left\{ 1, \frac{w(\B{x}',\B\theta')}{w(\B{x},\B\theta)} \right\}.\]
\end{prop}
\begin{proof}
First, note that  $\hat\rho_1(\B{\theta}';\B{\theta})$, defined in \cref{eqn:DAOneBlockRatio}, can be expressed as $\min\{ 1,\gamma\}$, where
\[
\gamma \equiv \frac{\hat\pi(\B{\theta}'\mid\B{b})r( \B{\theta}\mid \B{\theta}')}{\hat\pi(\B{\theta}\mid\B{b})r(\B{\theta}'\mid \B{\theta})} .
\]
Similarly, $\hat\rho_1(\B{\theta};\B{\theta}') =  \min\{ 1,\gamma^{-1}\}$. But  if $\gamma \neq 0$ then  $\min\{1,\gamma\}/\min\{1,\gamma^{-1}\} = \gamma$. Therefore, the ratio  $\hat\rho_1(\B{\theta}';\B{\theta})/\hat\rho_1(\B{\theta};\B{\theta}')$ simplifies to $\gamma$, from which it follows that
\[
\frac{q(\B{x},\B{\theta}\mid \B{x}',\B{\theta}')}{q(\B{x}',\B{\theta}'\mid \B{x},\B{\theta})} = \frac{\hat\pi(\B{x}\mid \B\theta,\B{b})\hat\pi(\B{\theta}'\mid\B{b})}{\hat\pi(\B{x}'\mid \B\theta',\B{b})\hat\pi(\B{\theta}\mid\B{b})},
\]
where $q$ is defined by~\cref{e_efftrans}. Thus, the ratio appearing in~\cref{eqn:DAOneBlockRatio2} simplifies to
\begin{align*}
	 \frac{\pi(\B{x}',\B{\theta}' \mid \B{b})q(\B{x},\B{\theta}\mid \B{x}',\B{\theta}')}{\pi(\B{x},\B{\theta} \mid \B{b})q(\B{x}',\B{\theta}'\mid \B{x},\B{\theta}) } = &  \frac{\pi(\B{x}',\B{\theta}' \mid \B{b})\hat\pi(\B{x}\mid\B{\theta},\B{b})\hat\pi(\B{\theta}\mid\B{b}) }{\pi(\B{x},\B{\theta} \mid \B{b})\hat\pi(\B{x}'\mid\B{\theta}',\B{b})\hat\pi(\B{\theta}'\mid\B{b})}  \\
= & \frac{\pi(\B{x}',\B{\theta}' \mid \B{b})\hat\pi(\B{x},\B{\theta}\mid\B{b}) }{\pi(\B{x},\B{\theta} \mid \B{b})\hat\pi(\B{x}',\B{\theta}'\mid\B{b})}.
\end{align*}
{Therefore, the acceptance ratio at the  second stage is
\begin{equation}
\label{eqn:DAOneBlockRatio2b}
\rho_3(\B{x}',\B{\theta}';\B{x},\B{\theta}) = \min\left\{1,\frac{\pi(\B{x}',\B{\theta}' \mid \B{b})\hat\pi(\B{x},\B{\theta}\mid\B{b}) }{\pi(\B{x},\B{\theta} \mid \B{b})\hat\pi(\B{x}',\B{\theta}'\mid\B{b})} \right\}.
\end{equation}
From this equation is clear that we only need to focus on the ratio of the exact to approximate posterior densities $ \pi(\B{x},\B{\theta} \mid \B{b})/\hat\pi(\B{x},\B{\theta}\mid\B{b})$ which simplifies according to~\cref{e_inter1}. Moreover, it is straightforward to verify that $M_1(\B\theta)=\hat\pi(\B\theta\mid\B{b})/\pi(\B\theta\mid\B{b})$, where $M_1(\B\theta)$ is defined in \cref{eqn:nt}. Substituting these two results into \cref{eqn:DAOneBlockRatio2b} and recalling that $z(\B{x},\B\theta)=w(\B{x},\B\theta)/M_1(\B\theta)$, we obtain the desired result.}

\end{proof}

In this analysis, \cref{eqn:DAOneBlockRatio2b} shows that~\cref{ALG:OneBlockApproxDA} amounts to generating a proposal from the approximate posterior distribution $\hat\pi(\B{x},\B\theta\mid\B{b})$ in a Metropolis-Hastings independence sampler. Compared to~\cref{ALG:OneBlockApprox}, we expect~\cref{ALG:OneBlockApproxDA} to have lower statistical efficiency but with much improved computational efficiency, since the posterior distribution $\pi(\B{x},\B\theta\mid\B{b})$ is only evaluated when a ``good'' sample has been drawn.


{\cref{p_twostage} states that the acceptance ratio at the second stage is close to $1$} if $\hat\pi(\B{x},\B\theta\mid\B{b})$ is a good approximation to $\pi(\B{x}, \B\theta\mid\B{b})$. The value of the ABDA algorithm is seen by observing that regardless of the choice of proposal distribution for $\B{\theta}$, the acceptance rate at the second stage will be high when $\hat{\pi}(\B{x}, \B{\theta} \mid \B{b})$ closely approximates the true target density $\pi(\B{x}, \B{\theta} \mid \B{b})$. For instance, if a poor proposal density is used for $\B{\theta}$, then many bad states $\B{\theta}^\prime$ will likely be proposed, but they will be discarded without wasting the computational effort to evaluate $\pi(\B{x}^\prime, \B{\theta}^\prime \mid \B{b})$. Provided we construct $\hat{\pi}(\B{x}, \B{\theta} \mid \B{b})$ to closely approximate $\pi(\B{x}, \B{\theta} \mid \B{b})$ for all $(\B{x}, \B{\theta}) \in \text{supp} ~\pi$, we can be confident that the true density will only be evaluated for states that have a high chance of being accepted. Further, using the approximate density in the second stage reduces the problem of tuning a Metropolis(-Hastings) algorithm to one of tuning the proposal $r(\B{\theta}^\prime \mid \B{\theta})$, which is easier to do when $\B{\theta}$ is low-dimensional.

\subsection{Analysis of the Pseudo-marginal acceptance ratio}

We turn to the analysis of the Pseudo-marginal algorithm (\cref{ALG:PseudoMarginal}). It is worth recalling that for $K=1$, the Pseudo-marginal algorithm reduces to AOB.
\begin{prop}\label{p_pseudo}
{In~\cref{ALG:PseudoMarginal}, the acceptance ratio simplifies to}
	\[ \rho_4(\B{\theta}'; \B{\theta}) =  \min\left\{1, \frac{z_K(\B{\theta}') \pi(\B{\theta}'\mid \B{b})r( \B{\theta}\mid\B{\theta}' )}{z_K(\B{\theta})\pi(\B{\theta}\mid \B{b})r(\B{\theta}'\mid \B{\theta})}\right\}, \]
	where $z_K(\B{\theta}) \equiv \frac1K \sum_{j=1}^Kz(\B{x}^j,\B\theta)$.
\end{prop}
\begin{proof}
	The proof is similar to \cref{p_accept1} and is omitted.
\end{proof}

More insight can be obtained by considering the asymptotic behavior of $z_K(\B{\theta})$. By~\cref{eqn:expw} and the strong law of large numbers, it follows that, for fixed $\B{\theta}$, $z_K(\B{\theta}) \stackrel{\text{a.s.}}{\rightarrow} 1$ as $K \rightarrow \infty$. The interpretation is that the Pseudo-marginal algorithm exhibits similar behavior as AOB with respect to one block. The difference is that, by taking $K \rightarrow \infty$, we attain almost sure convergence, as opposed to an average behavior. Further, by the Central Limit Theorem and the fact that $\B{x}^1, \ldots, \B{x}^K$ are independent for fixed $\B{\theta}$, we have that $\sqrt{K} (z_K(\B{\theta}) - 1) \stackrel{\mathcal{L}}{\rightarrow} \mathcal{N}(0,\sigma_z^2(\B{\theta}))$, where
\[
\sigma_z^2 (\B\theta) \equiv \text{Var}_{\hat\pi(\B{x}\mid \B\theta,\B{b})}[z(\B{x},\B\theta)] =  \frac{1}{M_1^2(\B\theta)}\left(\frac{1}{M_2(\B\theta)}-\frac{1}{M_1^2(\B\theta)}\right) < \infty, ~~\forall \B{\theta} ~(a.e.).  
\]
It follows from the Central Limit Theorem that $z_K(\B{\theta})$ is $\sqrt{K}-$consistent for 1 \cite{LehmannCasella98}; i.e.,
\begin{equation}\label{e_bigo}
z_K(\B\theta) - 1 = \mc{O}_p(K^{-1/2})
\end{equation}
as $K \rightarrow \infty$. However, it is important to observe that $\sigma_z(\B\theta)$ will be close to zero when the approximate distribution is close to the true distribution, in which case $z_K(\B\theta)$ will be close to one with high probability even for small $K$.

This analysis shows the advantage of the Pseudo-marginal approach, namely that we can still achieve desirable marginalization behavior in the presence of a poor approximation to the full conditional distribution of $\pi(\B{x},\B\theta\mid\B{b})$. If the approximation $\hat\pi(\B{x},\B\theta\mid\B{b})$ is not very good, one can set $K$ to be large to guarantee $z_K(\B\theta)\approx 1$, at the expense of drawing repeated realizations from the approximating distribution. Otherwise, if the approximation $\hat\pi(\B{x},\B\theta\mid\B{b})$ is good, $z_K(\B\theta)\approx 1$ even for small $K$. The trade-off is thus a large number of realizations from a poor approximation, or few (one in the case of AOB) realizations from a quality approximation. Regardless, if $z_K(\B\theta)\approx 1$, \cref{ALG:PseudoMarginal} will closely approximate the one-block algorithm.

\subsection{Approximate posterior distribution using low-rank approximation}\label{ss_lowrank}
We briefly review the low-rank approach for sampling from the conditional distribution $\pi(\B{x}\mid \B\theta,\B{b})$ { that was used in~\cite{brown2016computational}, and which previously appeared in \cite{Bui-ThanhGhattasMartinEtAl13,CalSom05,Chung13,Chung14,WilcoxGhattas,TC,IsaacPetraStadlerEtAl15}}. First, recall that $\condinv=\mu\B{A}^\top\B{A} +
\sigma \prior^{-1}$ and assume $\prior^{-1}=\B{L}^\top\B{L}$, so that
$$
\condinv= \B{L}^\top\left(\mu\B{L}^{-\top}\B{A}^\top\B{A}\B{L}^{-1} +
\sigma\B{I}\right)\B{L}.
$$
Next, define a rank$-k$ approximation as
\begin{equation}\label{e_lowrank}
\B{L}^{-\top}\B{A}^\top\B{AL}^{-1} \approx \B{V}_k\B\Lambda_k\B{V}^\top_k,
\end{equation}
where $\B\Lambda_k$ is diagonal with the $k$ largest eigvenvalues of $\B{L}^{-\top}\B{A}^\top\B{AL}^{-1}$ and $\B{V}_k$ contains orthonormal columns. Combining the previous two equations, we obtain the approximate covariance
$$
\condhinv \equiv ( \mu\B{L}^\top\B{V}_k\B\Lambda_k\B{V}_k^\top\B{L}+\sigma \prior^{-1}) =  \B{L}^\top(\mu\B{V}_k\B\Lambda_k\B{V}_k^\top+\sigma \B{I})\B{L}.
$$
Defining $\xcondh = \mu\condh\B{A}^\top\B{b}$, the approximate posterior distribution is obtained as
\begin{equation} \label{e_postapprox}
\widehat{\pi} (\B{x},\B\theta \mid \B{b}) \propto
f(\B\theta) \exp\left(-\frac{\mu}{2}\B{b}^\top \B{b} + \frac{\mu^2}{2}\B{b}^\top\B{A}\condh\B{A}^\top \B{b}  - \frac{1}{2} \|\B{x}-\xcondh\|_{\condhinv}^2 \right),
\end{equation}
where, $f(\B\theta) = \mu^{M/2}\sigma^{N/2}\pi(\B{\theta})$. The marginal density is
\begin{equation} \label{e_margapprox}
\hat\pi(\B{\theta}\mid \B{b}) \propto  \> \frac{f(\B\theta)}{\sqrt{\det(\condhinv)}} \exp\left(-\frac{\mu}{2}\B{b}^\top \B{b} + \frac{\mu^2}{2}\B{b}^\top\B{A}\condh\B{A}^\top \B{b}\right).
\end{equation}
To sample from the approximate conditional distribution, $\hat\pi(\B{\theta}\mid \B{b})$ we compute
\begin{equation}\label{e_condsample}
 \B{x} = \xcondh + \B{G\varepsilon} \qquad \B{\varepsilon} \sim \mathcal{N}(\B{0},\B{I}),
\end{equation}
where $\B{G} := \sigma^{-1/2} \B{L}^{-1} (\B{I} - \B{V}_k\widehat{\B{D}}_k\B{V}_k^\top)$ and $\widehat{\B{D}}_k = \B{I} \pm \left( \B{I}-\B{D}_k\right)^{1/2}$.
It is shown in~\cite{brown2016computational} that $\widehat{\B\Gamma}_{\text{cond}} = \B{GG}^\top$. Since $\widehat{\B{D}}_k$ is diagonal and $k \ll n$, \cref{e_condsample} provides a computationally cheap way of generating draws from $\widehat{\pi} (\B{x} \mid \B{b},\B\theta)$.

\paragraph{Acceptance ratio analysis}
From~\cite[Proposition 2]{brown2016computational}, the expression for $w(\B{x},\B\theta)$ simplifies considerably. First,
\begin{equation}\label{e_inter2}
\condinv -\condhinv = \mu\B{L}^\top\left(\sum_{j=k+1}^N\lambda_j\B{v}_j\B{v}_j^\top \right)\B{L},
\end{equation}
implying $w(\B{x},\B\theta) = \exp\left(-\frac12 \sum_{j =k+1}^N \lambda_j(\B{v}_j^\top \B{Lx})^2  \right)$, which in turn implies
\cite[Theorem 1]{brown2016computational}
\[ M_m(\B\theta) = \exp\left( \frac{\mu^2}{2\sigma}\sum_{j=k+1}^n \frac{m\mu\lambda_j}{m\mu\lambda_j + \sigma}(\B{b}^\top\B{A}\B{L}^{-1}\B{v}_j)^2\right)\prod_{j=k+1}^n \left( 1 + \frac{m\mu}{\sigma}\lambda_j\right)^{1/2}.\]
The key insight from this calculation is that $z(\B{x},\B\theta) = w(\B{x},\B\theta)/ M_1(\B\theta)$ only contains eigenvalues that are discarded in the low-rank approximation. This means if the discarded eigenvalues are very small, AOB is very close to the one-block algorithm. Similarly, for the ABDA algorithm, this implies the acceptance ratio of the second stage is very high and for the Pseudo-marginal MCMC, the variance $\sigma_z(\B\theta)$ is close to zero.

In practice, computing the truncated SVD can be computationally expensive. However, the low-rank decomposition can be approximately computed using a Krylov subspace approach~\cite{simon2000low}, or using a randomized approach~\cite{halko2011finding}. These approximate factorizations  can be used instead to construct the approximating distribution, such as the conditional $\widehat\pi(\B{x}\mid\B{b},\B\theta)$, the marginal $\widehat\pi(\B\theta\mid\B{b})$, and the full posterior distribution $\widehat\pi(\B{x},\B{b}\mid\B\theta)$. The details are given in~\cite{brown2016computational}.

\subsection{Computational cost} We briefly review the computational cost of the three algorithms proposed in this paper. Denote the computational cost of forming the matrix-vector product (henceforth, referred to as matvec) with $\B{A}$ by $T_{\B{A}}$. Similarly  denote the cost of a matvec with $\B{L}$ and $\B{L}^{-1}$ as $T_{\B{L}}$ and $T_{\B{L}^{-1}}$, respectively. To simplify the analysis, we make two assumptions: first, the cost of the matvecs dominates the computational cost, and second, the cost of the transpose operations of the respective matrices is the same as that of the original matrix.
\begin{table}[tb]\centering
\begin{tabular}{lll}
Component & Dominant Cost & Additional Cost \\ \hline
	Pre-computation & $2k(T_{\B{A}} + T_{\B{L}^{-1}}  )$ & $\mathcal{O}(nk^2)$ \\
Sample & $T_{\B{L}^{-1}}$ & $\mathcal{O}(nk)$ \\
Acceptance ratio & $T_{\B{A}} + T_{\B{L}}$ & $\mathcal{O}(nk)$.
\end{tabular}
\caption{Summary of various computational costs in the Approximate One-Block Algorithm}
\end{table}

\paragraph{AOB algorithm} We first analyze the computational cost associated with AOB (\cref{ALG:OneBlockApprox}). There are three major sources of computational cost. In the offline stage, we pre-compute a low-rank approximation as in~\cref{e_lowrank}. By using a Krylov subspace solver or randomized SVD algorithm, the dominant cost is in computing the matvecs, i.e., $2k(T_{\B{A}} + T_{\B{L}^{-1}} )$ floating point operations (flops), with an additional cost of $\mathcal{O}(nk^2)$ flops. In the online stage,  to generate a composite sample $(\B{x}',\B\theta')$, the major cost that depends on the dimension of the problem is due to~\cref{e_condsample} and can be quantified as $T_{\B{L}^{-1}}$ flops, with an additional $\mathcal{O}(nk)$ cost. Finally, computing the acceptance ratio requires one evaluation of the full posterior distribution, and one evaluation of the approximate posterior distribution. Evaluating the full posterior using~\cref{e_post2} requires $T_{\B{A}} + T_{\B{L}}$ flops, with an additional cost $\mathcal{O}(nk)$ flops, whereas evaluating the approximate posterior distribution using~\cref{e_postapprox} only requires $\mathcal{O}(nk)$ flops.

\paragraph{ABDA algorithm} The cost of ABDA, \cref{ALG:OneBlockApproxDA}, is the same as that of AOB, \cref{ALG:OneBlockApprox}, with only one notable exception. The extra step of evaluating the approximate marginal distribution~\cref{e_margapprox} requires an additional $2T_{\B{L}^{-1}} + \mathcal{O}(nk)$ flops. However, while the per-iteration cost of ABDA is comparable to that of AOB, the overall cost of ABDA can be lower. The reason is that the acceptance ratio involving the full posterior distribution is only evaluated in the second stage, unlike AOB in which the posterior distribution is evaluated at every iteration. The computational speedup has been demonstrated in numerical experiments, see \cref{sec:num}. Estimates of the computational speedup can be obtained by following the approach in~\cite{cui2011bayesian}, but we omit this discussion.

\paragraph{Pseudo-marginal algorithm} The cost of the Pseudomarginal approach (\cref{ALG:PseudoMarginal}) is the same as~\cref{ALG:OneBlockApprox} {with two exceptions:  $K$ samples need to be generated which costs $K [ T_{\B{L}^{-1}} + \mathcal{O}(nk) ]$ flops, and to evaluate the acceptance ratio, we need $K[T_{\B{A}} + T_{\B{L}} + \mathcal{O}(nk)]$ flops. The benefits of the Pseudo-marginal algorithm, i.e., more effective marginalization, has to be weighed against the additional cost of generating the samples and evaluating the acceptance ratio.}

\section{Numerical Experiments}\label{sec:num}
 Throughout this section, the proposal distribution for $\B{\theta}$ is taken to be the Adaptive Metropolis proposal~\cite{haario2001adaptive} using a lognormal proposal for $\B\theta$. For the prior and noise precision parameters, we assign a Gamma prior, $\text{Gamma}(1, 10^{-4})$, based on the recommendation of~\cite{bardsley2013efficient}.  If {\em a priori} knowledge is available regarding the distribution of $\B\theta$, this can be incorporated into the hyperpriors. It is also worth mentioning that besides the Gamma priors, other choices for the hyper-priors can also be made, e.g., proper Jeffreys. See~\cite{brown2016computational} for a detailed discussion. Besides these parameters, no other tuning parameters were necessary.

\subsection{1D Deblurring}\label{ssec:1d}
We take a simple 1D deblurring example to illustrate our algorithms. {This application arises from the discretization of a Fredholm integral equation of the first kind as in \cref{eqn:fred}. The details of this application are given elsewhere and we refer the reader to them~\cite{Bar11}. For the problem size we choose $N=128$ and add  Gaussian noise with variance $0.01^2\|\B{Ax}_\text{true}\|_{2}^2$ to simulate measurement error, where $\B{x}_\text{true}$ is the true vector which is assumed to be known. We model smoothness on $\B{x}$ {\em a priori} by taking $\Gamma_\text{prior}^{-1} = -\Delta $, where $-\Delta$ is the discrete Laplacian with Dirichlet boundary conditions~\cite{KaiSom05}.}

A low-rank approximation with $k=35$ computed using the ``exact'' SVD was used to define the approximate posterior distribution. We run $20,000$ iterations of the AOB algorithm. The first $10,000$ samples were considered to be the burn-in period and were discarded. The left panel of~\cref{fig:oned} shows the true vector $\B{x}_\text{true}$ and the blurred vector $\B{b}$ with the simulated measurement noise. The right panel of the same figure shows the approximate posterior mean superimposed on the true vector $\B{x}_\text{true}$; also plotted are the curves corresponding to the pointwise $95\%$ credible bounds. The plots show that the true image is mostly contained within the credible intervals, thereby providing a measure of confidence in the reconstruction process.

\begin{figure}[tb]\centering
\includegraphics[scale=0.3]{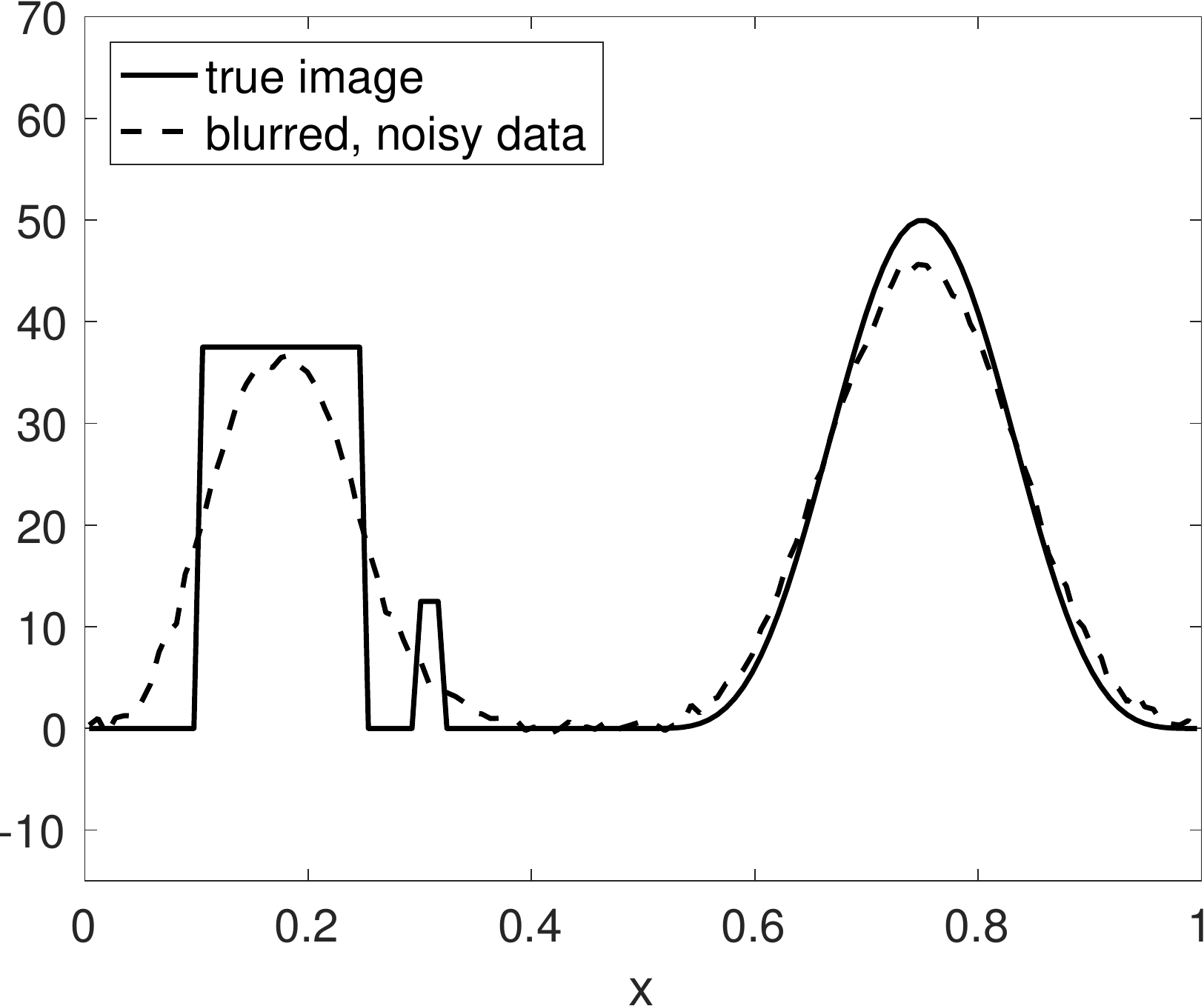}
\includegraphics[scale=0.3]{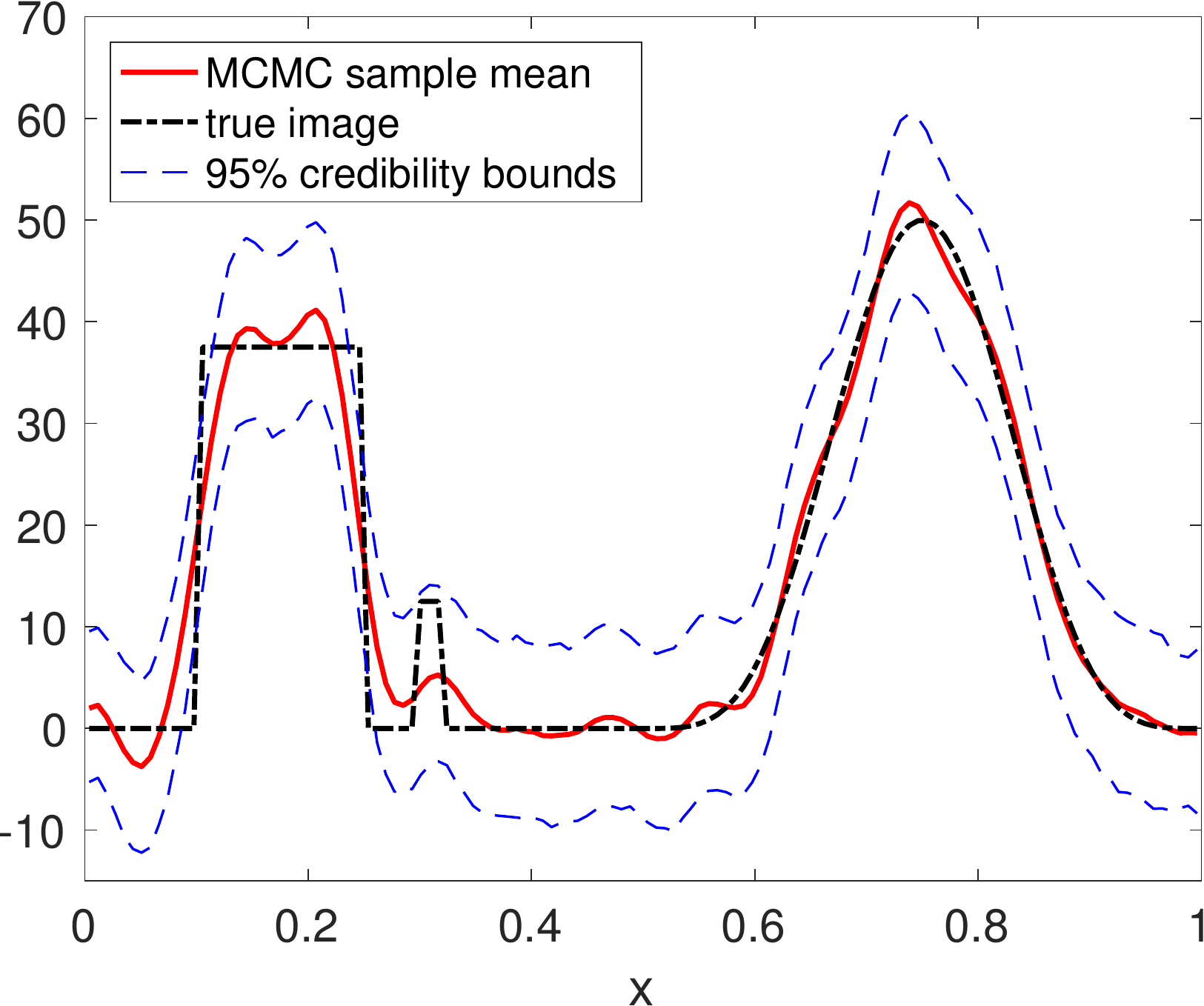}
\includegraphics[scale=0.3]{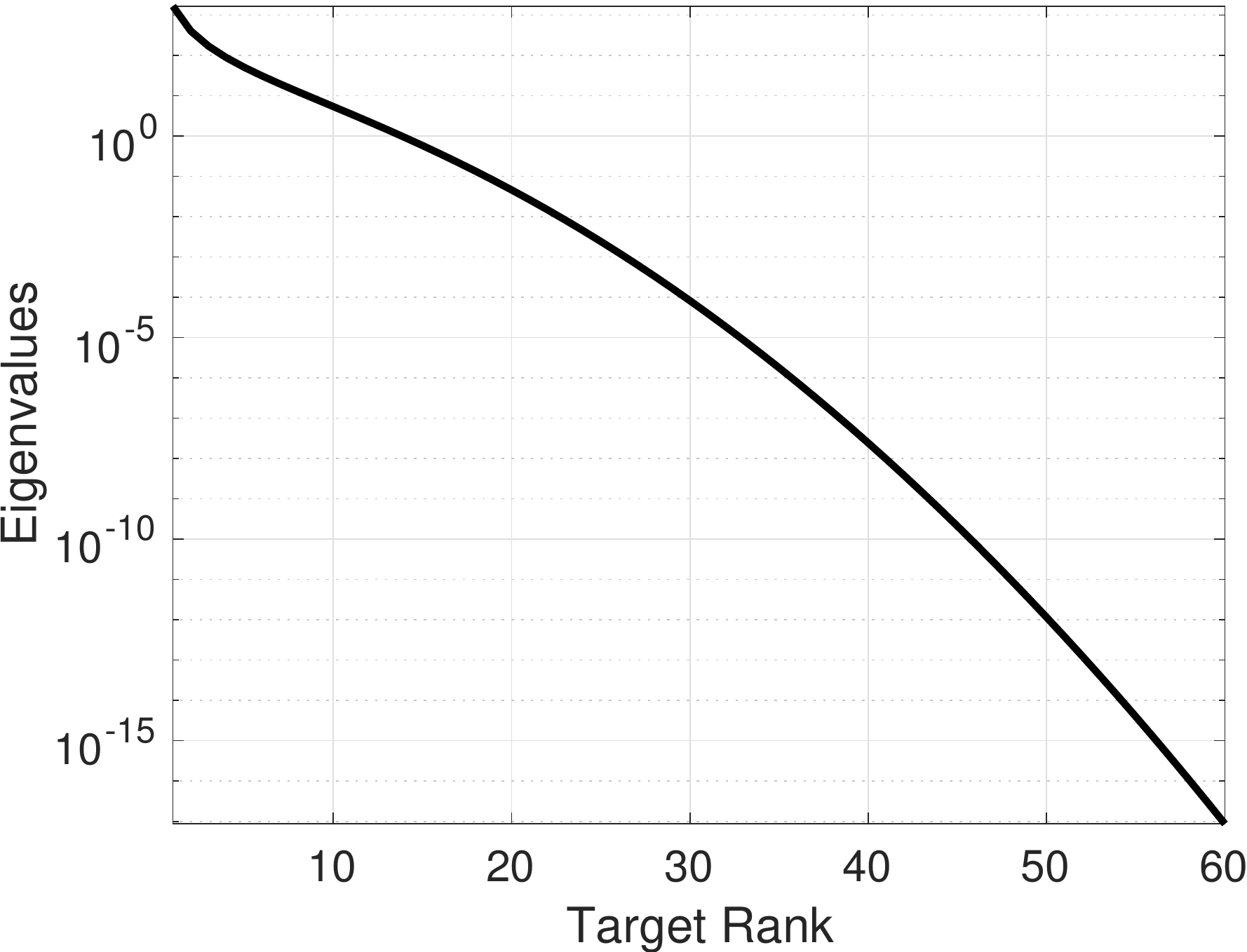}
\caption{(left) The true solution $\B{x}_\text{true}$ and the blurred vector $\B{b}$ with $2\%$ Gaussian noise to simulate measurement error. (center) The posterior mean superimposed with the pointwise $95\%$ credible bounds. (right) The spectrum of the matrix $\B{L}^{-\top}\B{A}^\top\B{A}\B{L}^{-1}$.}
\label{fig:oned}
\end{figure}

\paragraph{Diagnostics for convergence} We now provide some limited diagnostics to assess the convergence of the algorithm. In the upper left panel of~\cref{fig:oneddiag}, we display the trace plots of the precision parameters $\mu$ and $\sigma$ for three different chains with different initializations. It is seen that the chains appear to converge. Furthermore, the Multivariate Potential Scale Reduction Factor (MSPRF) \cite{BrooksGelman98} for the $\B{x}$ chain is $1.06$, which is within the recommend rule-of-thumb range (less than $1.1$). The top right panel of the same figure shows the histograms of the precision parameters. The (two-sided) Geweke test~\cite{bardsley2018} applied to the resulting  $\mu$- and  $\sigma$-chains yield $p$-values both greater than $0.98$, meaning that there is no reason to believe the chains are out of equilibrium. Further support for convergence is obtained by monitoring the cumulative averages of the $\mu$ and $\sigma$-chains plotted in the bottom panel of~\cref{fig:oneddiag}, from which it is readily seen that the cumulative averages converge. 
\begin{figure}[tb]\centering
\includegraphics[scale=0.3]{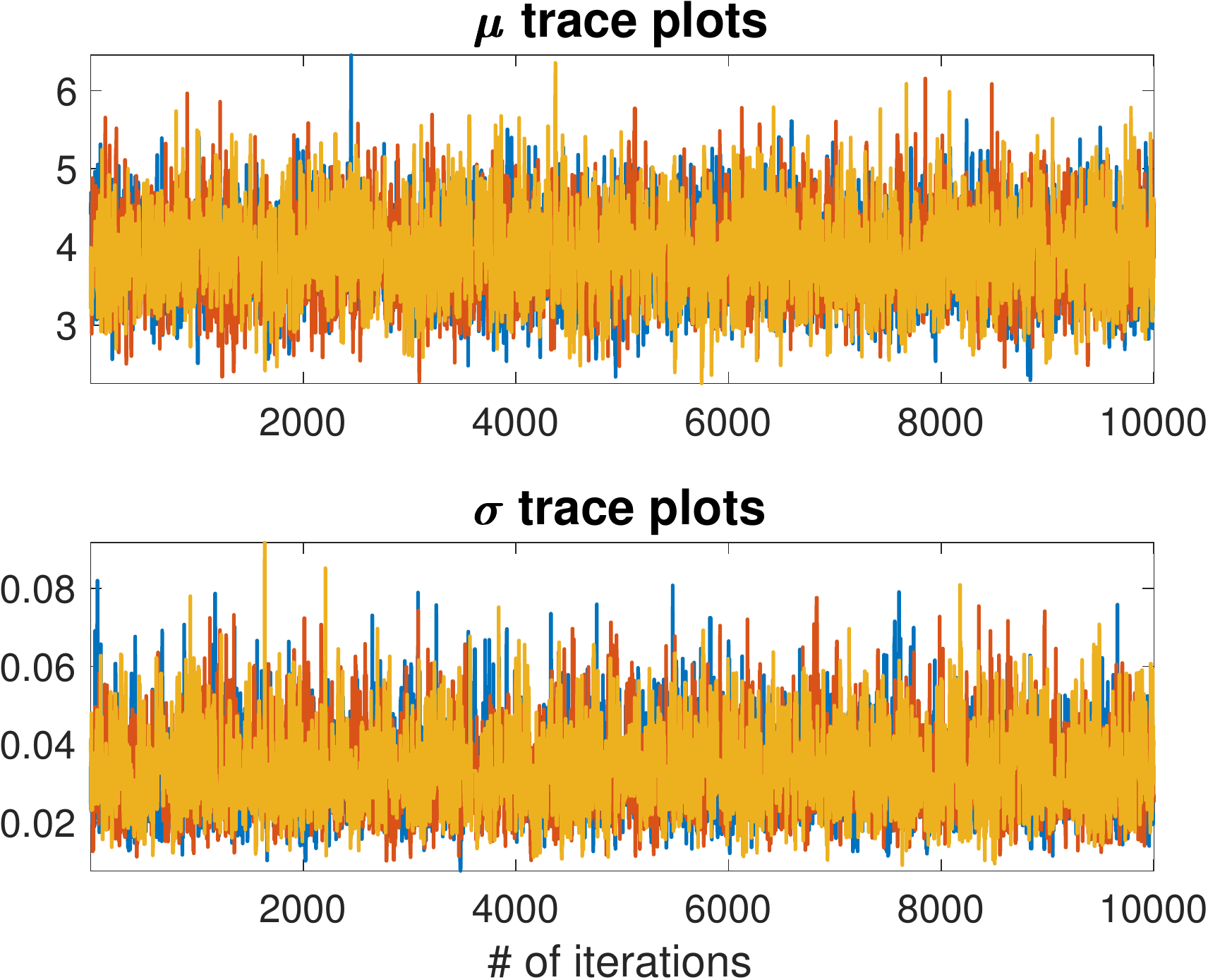}
\includegraphics[scale=0.3]{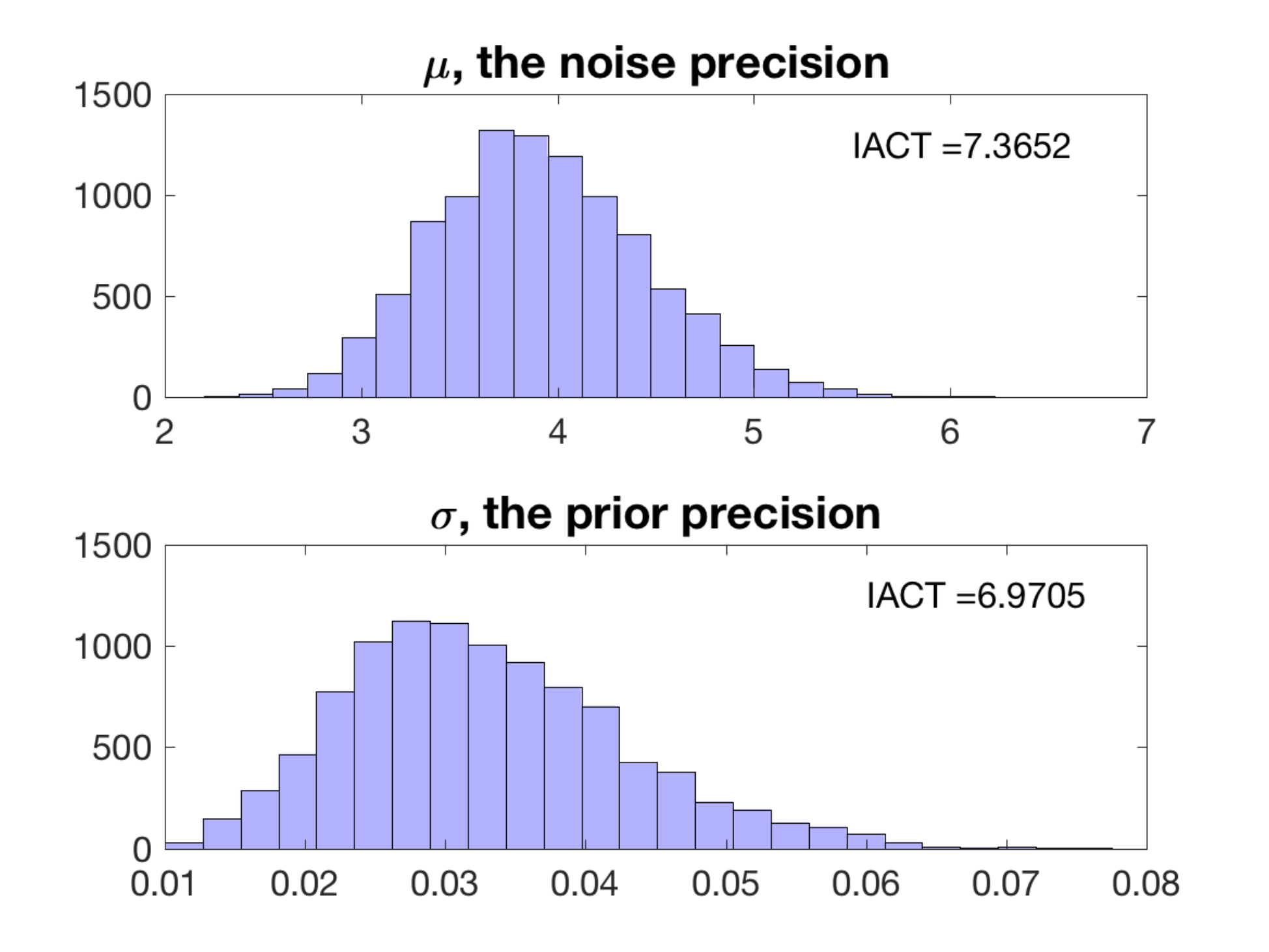}
\includegraphics[scale=0.4]{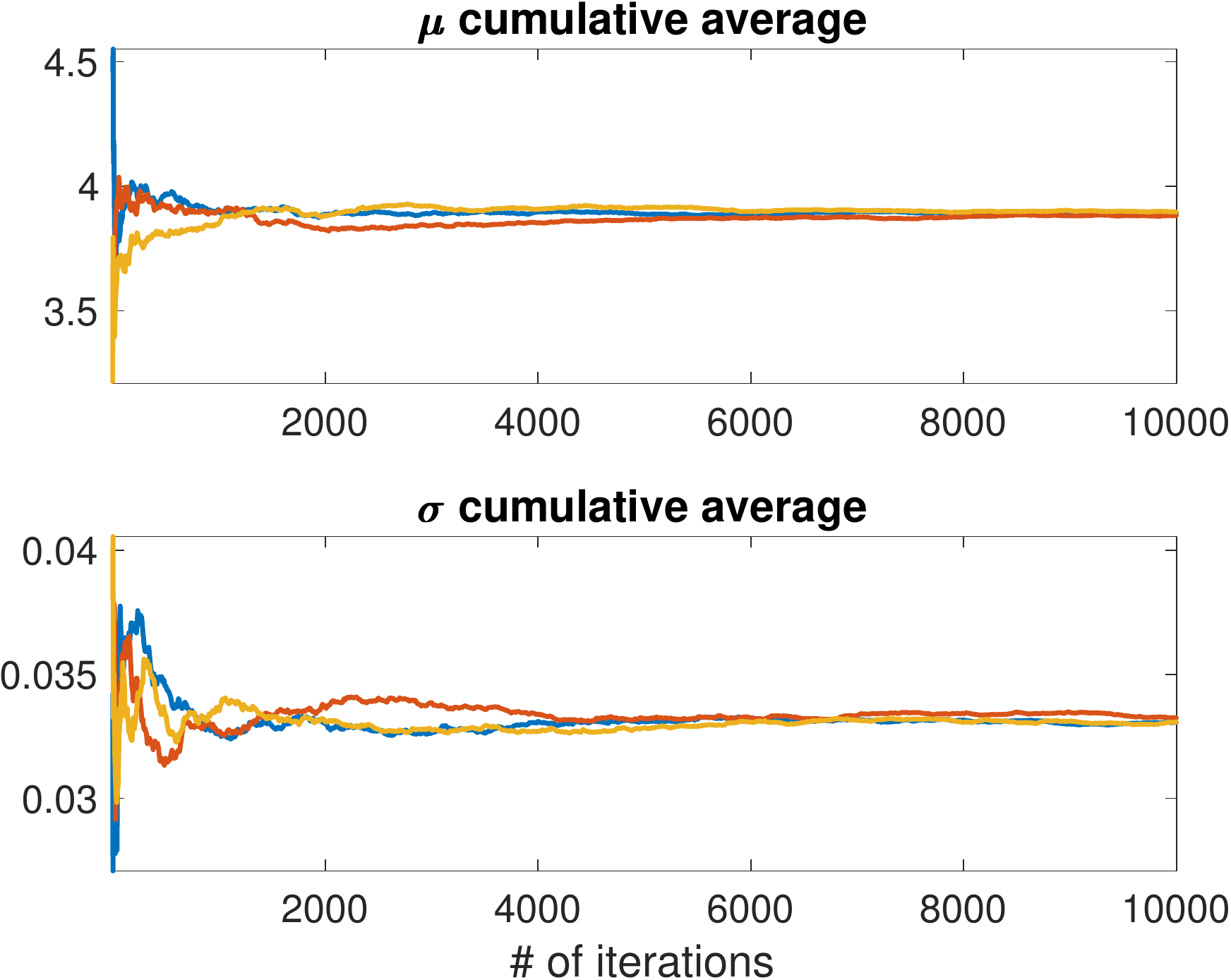}
\caption{ The (top left) trace plots and (top right) histograms of the $\mu$ and $\sigma$ chains for the $1$-D deblurring example using the AOB algorithm. (Bottom) Cumulative averages of the $\mu$ and $\sigma$ chains against the number of iterations. The target rank is $k=50$.}
\label{fig:oneddiag}
\end{figure}

In the right panel of~\cref{fig:ACF} (in~\cref{sec:review}), we show the empirical autocorrelation function of one of the $\sigma$ chains. As mentioned earlier, the autocorrelation decays rapidly. While we do not plot the autocorrelation of the $\mu$ chain, we observed that it had similar behavior. The integrated autocorrelation time (IACT) was $6.97$ for the $\sigma$-chain. The acceptance rate was $33.8\%$ and the effective sample size (ESS) was $1434.61$. Here, the ESS \cite{KassEtAl98} is defined as {ESS} $= {N_s}/\tau_\sigma$, where $N_s$ is the Monte Carlo sample size and $\tau_\sigma$ denotes IACT of the Markov chain.

\paragraph{Comparing different methods} We now compare all three algorithms proposed in~\cref{sec:alg}. The setup of the problem is the same as the previous experiment. The previous analysis shows that the AOB algorithm does a good job decorrelating the  $\sigma$-chain from the $\B{x}$ chain. We repeat this experiment for the other two methods as well.~\cref{fig:ACF2} plots the autocorrelation functions of the $\sigma$-chains produced by the other two methods, ABDA and PM with $K=5$. As is readily seen, both methods successfully reduce the autocorrelation in the $\sigma$-chains. Although we do not plot the $\mu$-chains, they behave similarly to the $\sigma$-chains.
\begin{figure}[tb]\centering
\includegraphics[scale=0.4]{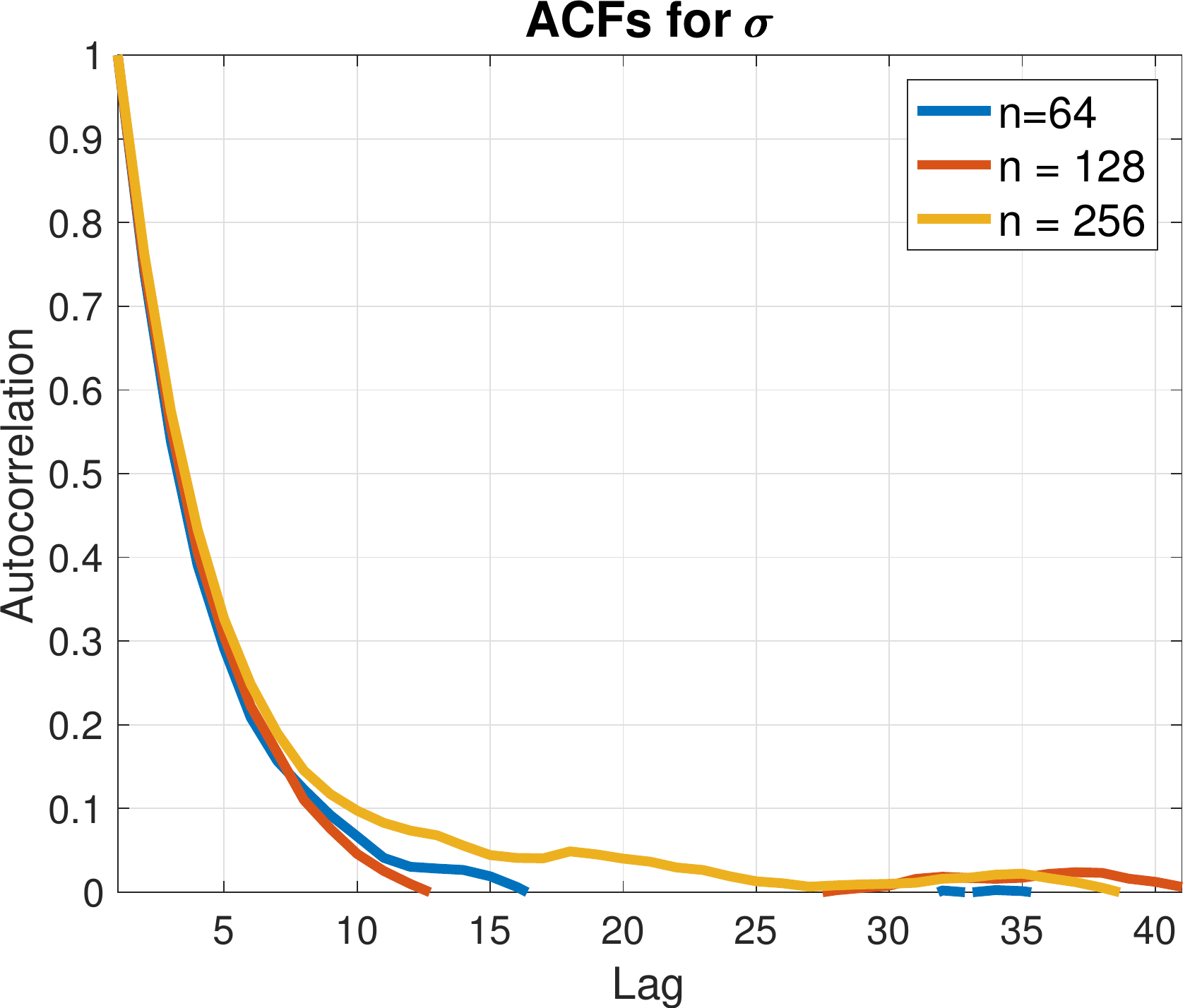}
\includegraphics[scale=0.4]{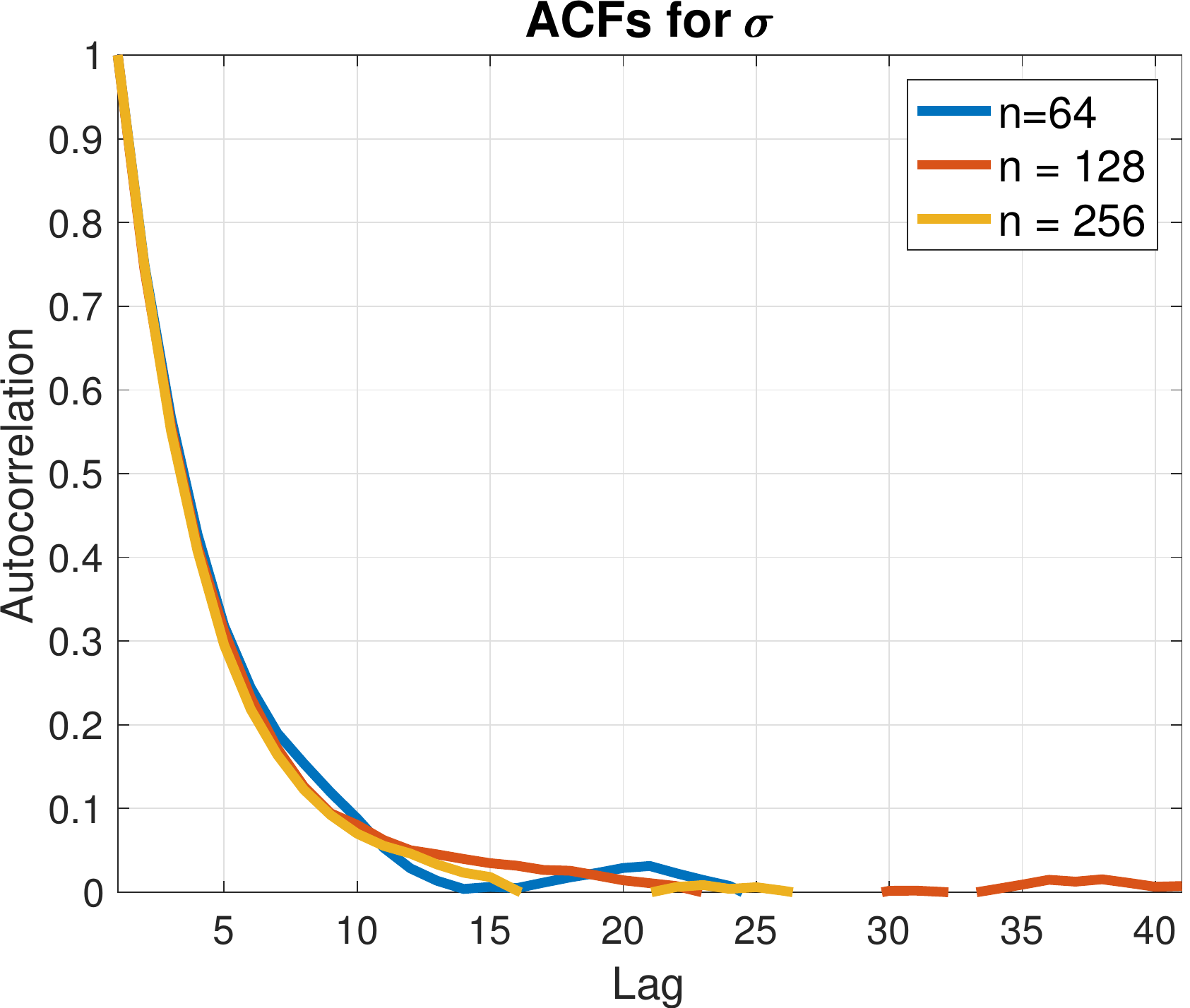}
\caption{Plots of the empirical autocorrelation functions for the $\sigma$ chains in the $1$-D deblurring example: (left) ABDA and (right) PM with $K=5$. }
\label{fig:ACF2}
\end{figure}


In~\cref{tab:1dsum}, we compare the various methods proposed here. The ESS of the PM algorithm is  marginally better than ABDA or AOB. This is because PM has a lower IACT. This is not surprising since the PM uses additional samples to marginalize the posterior distribution. More details regarding the PM method are given in the next experiment. Between the AOB and ABDA algorithms, AOB has a higher ESS and is more statistically efficient. In principle, though, ABDA is much less computationally expensive since the forward operator $\B{A}$ is evaluated less frequently. The difference is more pronounced in the inverse heat equation example in which the forward model is more expensive to evaluate.

\begin{table}[tb]\centering
\begin{tabular}{c|c|c}
& IACT $\sigma$ & ESS  \\ \hline
AOB & $6.97$ & $1434.61$  \\
ABDA &   $7.14$  & $1400.90$  \\
PM $K=5$ & $6.90$ & $1448.38$
\end{tabular}
\caption{Integrated autocorrelation times (IACT) and effective sample sizes (ESS) for the $\sigma$ chains produced by the proposed algorithms for the $1$-D deblurring example.}
\label{tab:1dsum}
\end{table}

\paragraph{Effect of target rank} Here we investigate the effect of the target rank on the acceptance rate and the statistical efficiencies of the samplers. The model is identical to that in the first experiment with the exception that we change the target rank $k$ that controls the accuracy of the approximate posterior distribution. We vary the target rank $k$ from $20$ to $40$ in increments of $5$. The results are displayed in~\cref{tab:1drank}. When the target rank $k$ is below $25$, we find that the acceptance rate is very close to zero, indicating that the low-rank approximation is not sufficiently accurate. On the other hand, when the target rank $k$ is $25$ and above, the acceptance rate is close to $33\%$. It is also seen that the increasing the target rank does not substantially increase the acceptance rate. {This suggests that we have successfully approximated the posterior distribution to the point that the only limitation is the proposal distribution for $\B\theta$; effectively, the dimensionality has been reduced from $128$ to $25$ which is roughly a factor of $5$.} Similar results are observed for the ABDA algorithm as well, shown in the right panel of~\cref{tab:1drank}.

\begin{table}[tb]\centering
    \begin{subtable}{.5\linewidth}
      \centering
\begin{tabular}{c|c|c}
Rank & ESS & A.R. \\ \hline
$20$ & $259.38$ & $0.1123$  \\
$25$ & $1353.71$ & $0.3318$  \\
$30$ & $1480.83$ & $0.3383$ \\
$35$ & $1357.73$ & $0.3396$ \\
$40$& $1357.73$ & $0.3396$
\end{tabular}
\caption{AOB}
\end{subtable}%
\begin{subtable}{.5\linewidth}
      \centering
\begin{tabular}{c|c|c|c}
Rank & ESS & A.R. 1 & A.R. 2\\ \hline
$20$ & $7.67$ & $0.0048$ & $0.1465$  \\
$25$ & $1056.31$ & $0.4988$ & $0.7017$   \\
$30$ & $1326.67$ & $0.5024$ & $0.7242$\\
$35$ & $1139.74$ & $0.4997$ & $0.7264$\\
$40$& $1139.74$ & $0.4997$ & $0.7264$
\end{tabular}
\caption{ABDA}
 \end{subtable}
    \caption{Summary of the effect of increasing rank for the $1$-D deblurring example. `A.R.' refers to acceptance rate; the number as a suffix refers to the stage. In general, an increase in the target rank increases the acceptance ratio and the ESS. }
\label{tab:1drank}
\end{table}
\paragraph{Effect of increased importance sample size} In this experiment, we explore the effect of the importance sample size $K$ on the statistical efficiency in the PM algorithm. In addition to the ESS and the IACT of the $\sigma$-chain, we report the CPU time in seconds, and the Computational cost per Effective Sample (CES), which is the ratio of the CPU time to the ESS.~\cref{p_pseudo} shows that there are two ways to make the PM algorithm closer to one-block: by increasing the target rank that defines the approximate distribution, or by increasing the number of samples $K$. We take the target rank to be $k=20$. From the previous experiment it is seen that the samples have high autocorrelation and the resulting ESS is poor.~\cref{tab:pm} shows that by increasing the number of samples $K$, the PM approach seems to be performing better as evident in the increase in the ESS. However, the CPU time also increases considerably since each step in the algorithm is more expensive, with increasing $K$. Indeed, when $K$ increases by a factor of $10$, the ESS merely doubles. On the other hand, from~\cref{tab:1drank} (corresponding to PM with $K=1$), increasing the target rank $k$ by $10$ substantially increases the ESS, but is far less expensive. However, when the spectrum of the prior-preconditioned Hessian is flat, increasing the target rank may not improve the ESS. On the other hand,~\cref{p_pseudo} shows that increasing the sample size $K$ will have the desired effect.
\begin{table}[tb]\centering
\begin{tabular}{c|c|c|c|c}
$K $& ESS &  IACT $\sigma$ & CPU time [s] & CES \\ \hline
$1$ & $282.49$ &  $32.00$ & $43.62$ & $0.1544$\\
$5$ & $657.52$ &  $12.29$ & $54.56$ &  $0.0830$\\
$10$  & $1019.37$  & $9.81$ & $65.87$ &  $0.0646$\\
$50$  & $1329.79$  & $7.52$ & $154.52$ & $0.1162$
\end{tabular}
\caption{Effect of the importance sample size $K$ on the $\sigma$-chains produced by the Pseudo-marginal approach in the $1$-D deblurring example. The target rank is fixed to be $k=20$.}
\label{tab:pm}
\end{table}

\subsection{A PDE-based example}
\newcommand{\D}{\mathcal{D}}
In this application, we consider the inversion of the two-dimensional initial state in the heat equation.
Given the initial state $u_0$, we solve
\[
    \begin{aligned}
    u_t - \kappa \Delta u &= 0, \quad \text{in } \D \times [0, T], \\
                        u(x, t) &= 0, \quad \text{on } \partial \D \times [0, T],\\
                        u(x, 0) &= u_0, \quad \text{in } \D.
    \end{aligned}
\]
The domain is taken to be $\D = [0, 2]\times [0, 1]$.  In the present experiment
the diffusion coefficient was set to $\kappa = .001$ and we used a final
simulation time of $T = 5$. The inverse problem seeks to use
point measurements of $u(\cdot, T)$ to reconstruct the initial state. The
assumed true initial state used to simulate the data is shown in~\cref{fig:heat_inv_setup}~(left).  We suppose that measurements are taken at an array of $512$
observation points depicted in~\cref{fig:heat_inv_setup}~(right).

 We used a
Gaussian prior with covariance operator $\mathcal{A}^{-2}$, where
$\mathcal{A}$ is the Laplacian with zero Dirichlet boundary condition.    The
data were generated by adding $10\%$ Gaussian noise to the measurements.  The
problem is discretized with a $64 \times 32$ grid. In the inverse problems, we seek to reconstruct a discretized
initial state in $\mathbb{R}^n$ with $n = 63 \times 31 = 1953$.
\begin{figure}[ht]
\includegraphics[width=.4\textwidth]{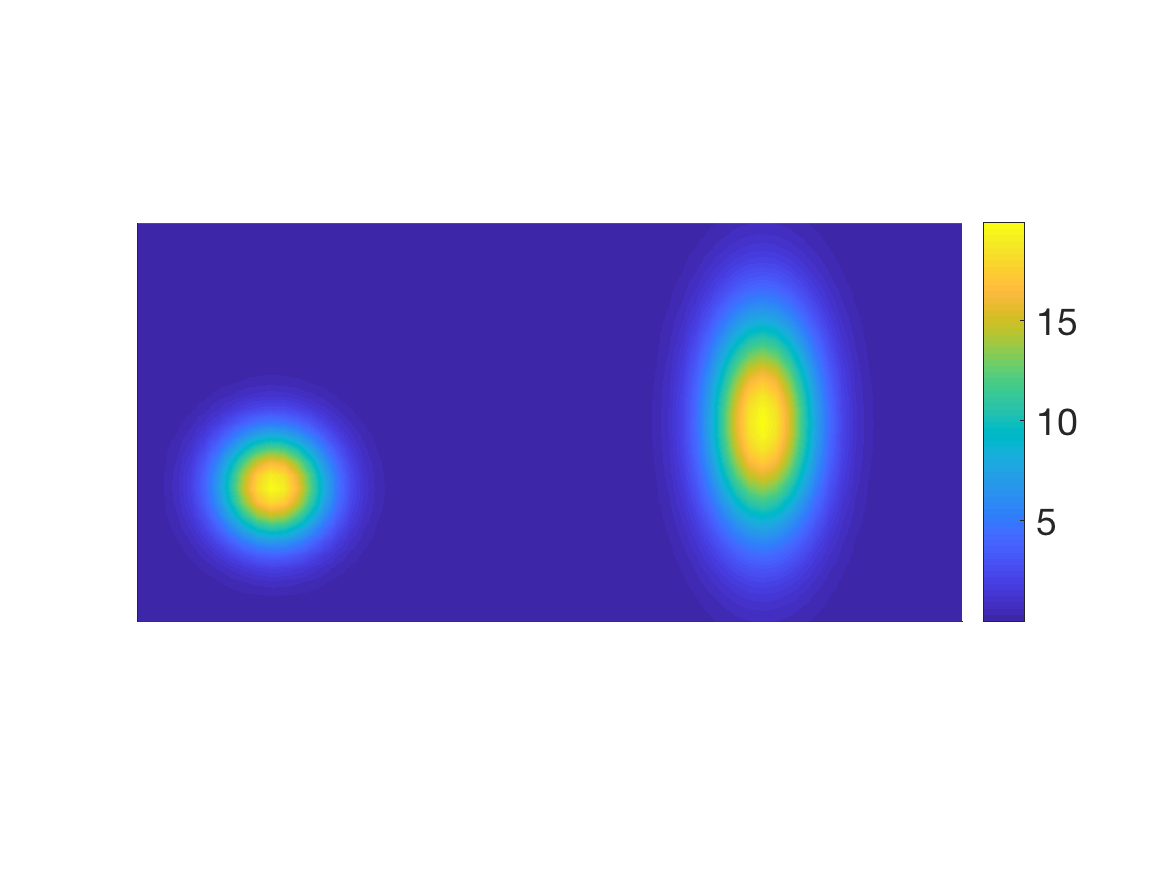}
\includegraphics[width=.4\textwidth]{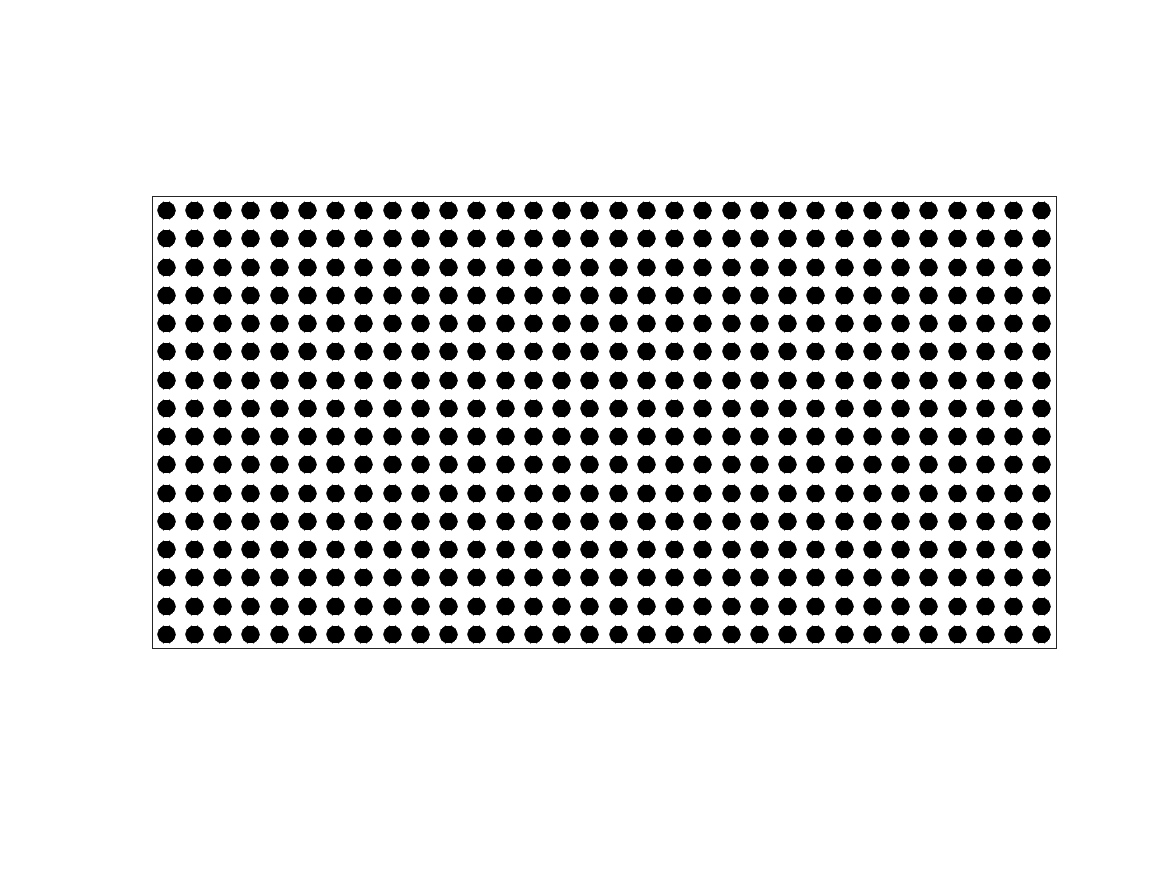}
\caption{True initial-state (left) and  observation points (right).}
\label{fig:heat_inv_setup}
\end{figure}

Let $\mathbf{A}$ denote the discretized forward operator that is obtained by composing
the PDE solution operator and the observation operator that extracts solution values at
the measurement points. Owing to the fast decay of singular values of $\mathbf{A}$, we can use
a low-rank approximation. In the present example, a rank-$70$ approximation
was found to provide sufficient accuracy.

\paragraph{Performance of AOB and ABDA}
As a first experiment, we apply the AOB method to the present inverse problem. Our MCMC implementation is run for $2\times 10^4$ iterations; we retain $10^4$ samples, and discard the rest as part of the burnin period.
In~\cref{fig:mcmc_approx_one_block}, we provide
the trace plots of the $\mu$ and $\sigma$ chains (left),
along with approximate posterior mean estimate of the initial state (right).
We also report the empirical autocorrelation functions for the $\mu$ and $\sigma$
chains in~\cref{fig:ACFs}. The acceptance rate for the algorithm was approximately 30\%.
\begin{figure}[!ht]
\begin{tabular}{cc}
\includegraphics[width=.4\textwidth]{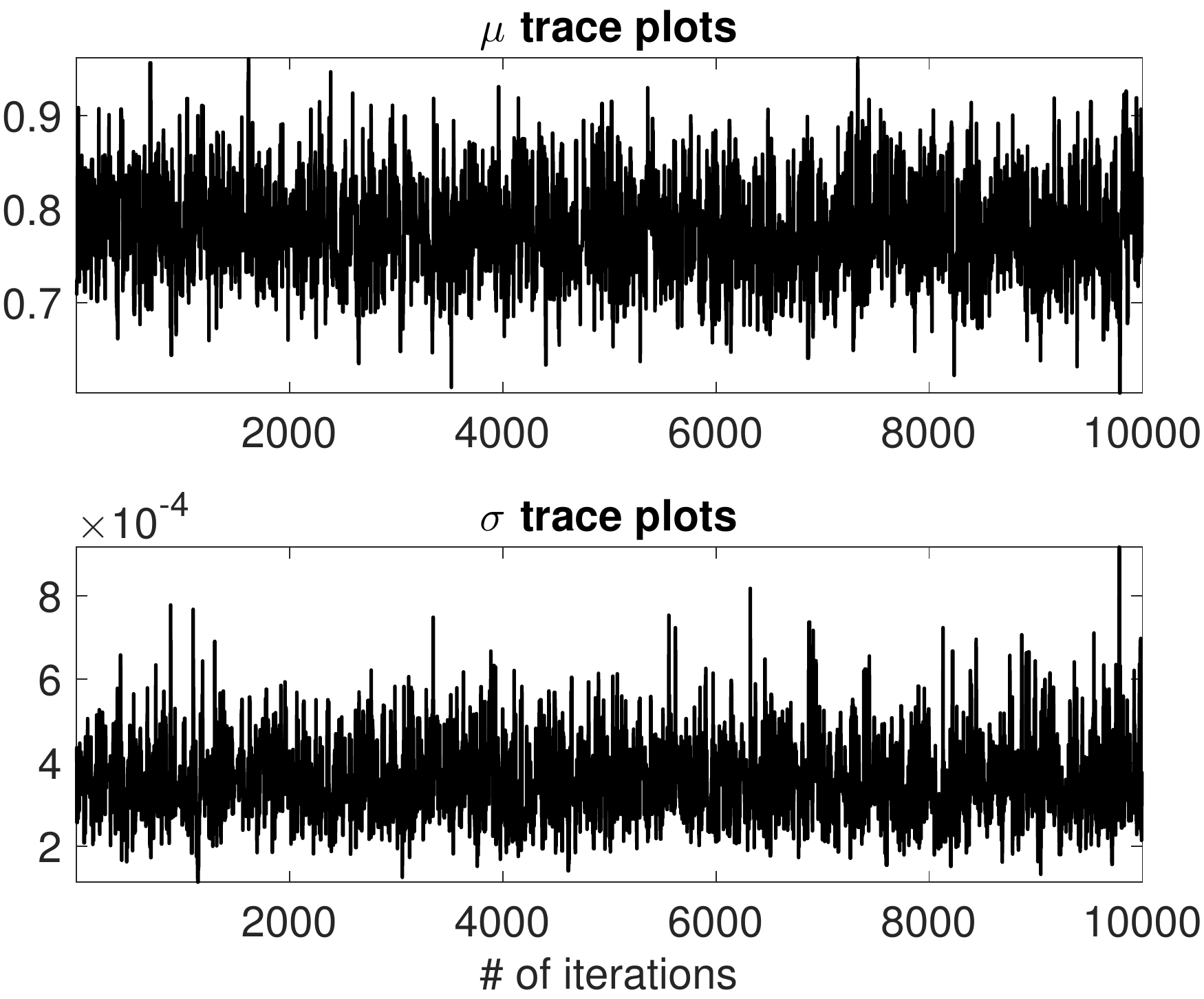}&
\includegraphics[width=.4\textwidth]{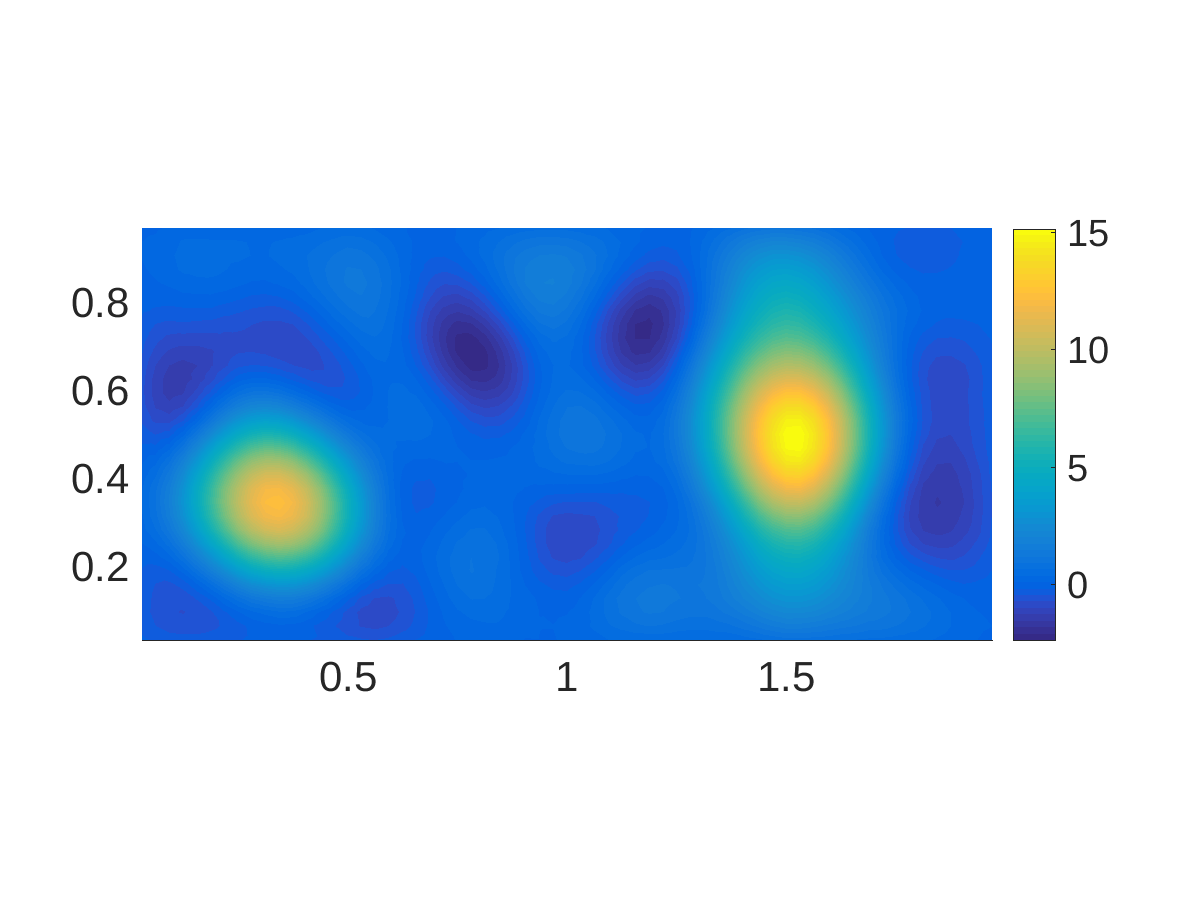}\\
\end{tabular}
\caption{AOB results for the inverse heat equation example.
Left: $\mu$ and $\sigma$ trace plots; right: Posterior mean estimate of
the initial state.
The Monte Carlo sample size was $10^4$.}
\label{fig:mcmc_approx_one_block}
\end{figure}

\begin{figure}[!ht]
\includegraphics[width=.45\textwidth]{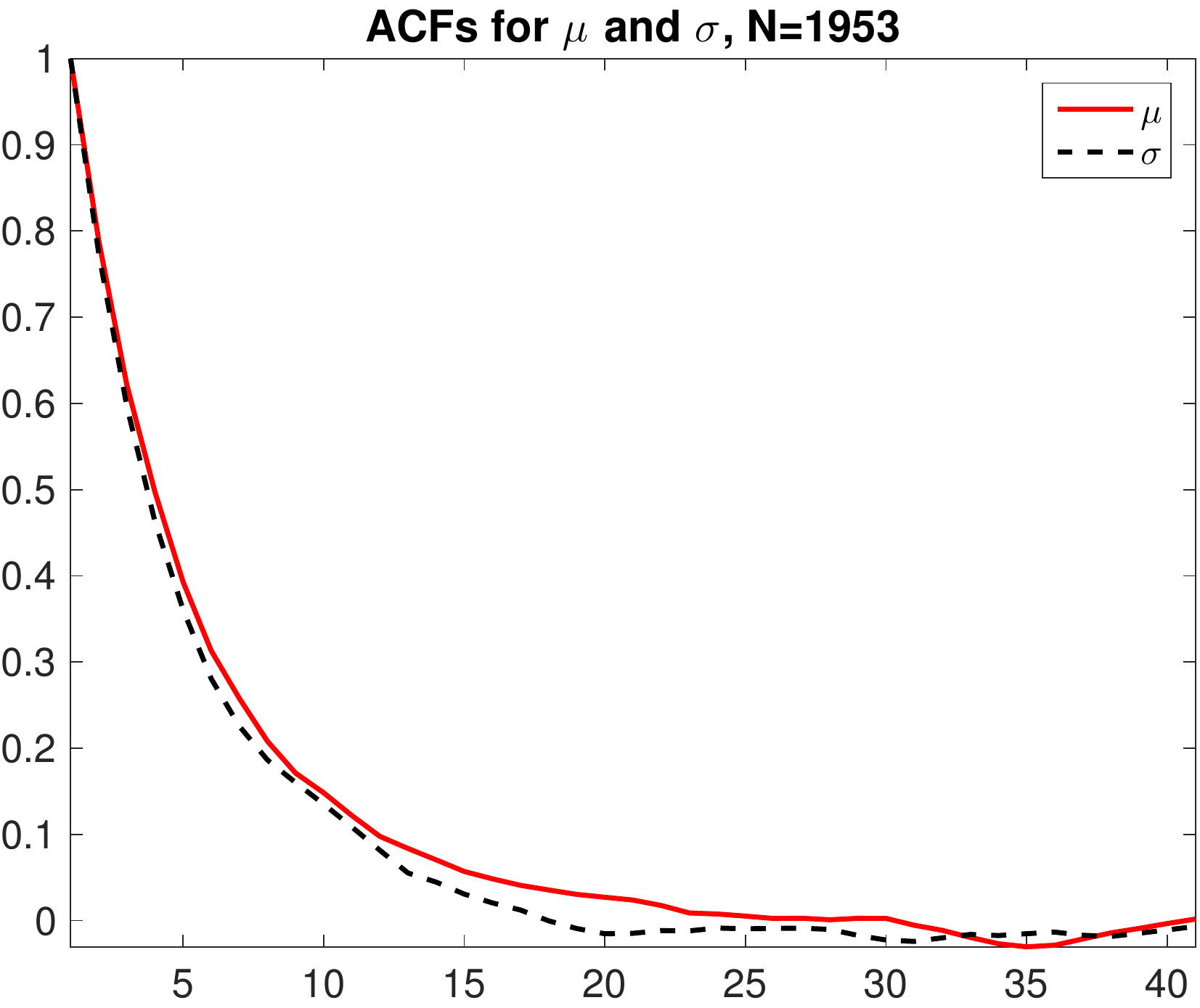}
\caption{The empirical autocorrelation function computed using AOB in the inverse heat equation example.}
\label{fig:ACFs}
\end{figure}

Next, we compare the performance of AOB with that of ABDA. The results of our
numerical experiments are summarized in \cref{fig:AOB_vs_ABDA}. We note that
both methods produce essentially the same results.
However, \cref{tbl:AOB_vs_ABDA} illustrates the reduction in computation time facilitated by ABDA. This application suggests that
ABDA may be preferable to AOB when the forward problem is computationally expensive.
\begin{figure}[!ht]
\includegraphics[width=.75\textwidth]{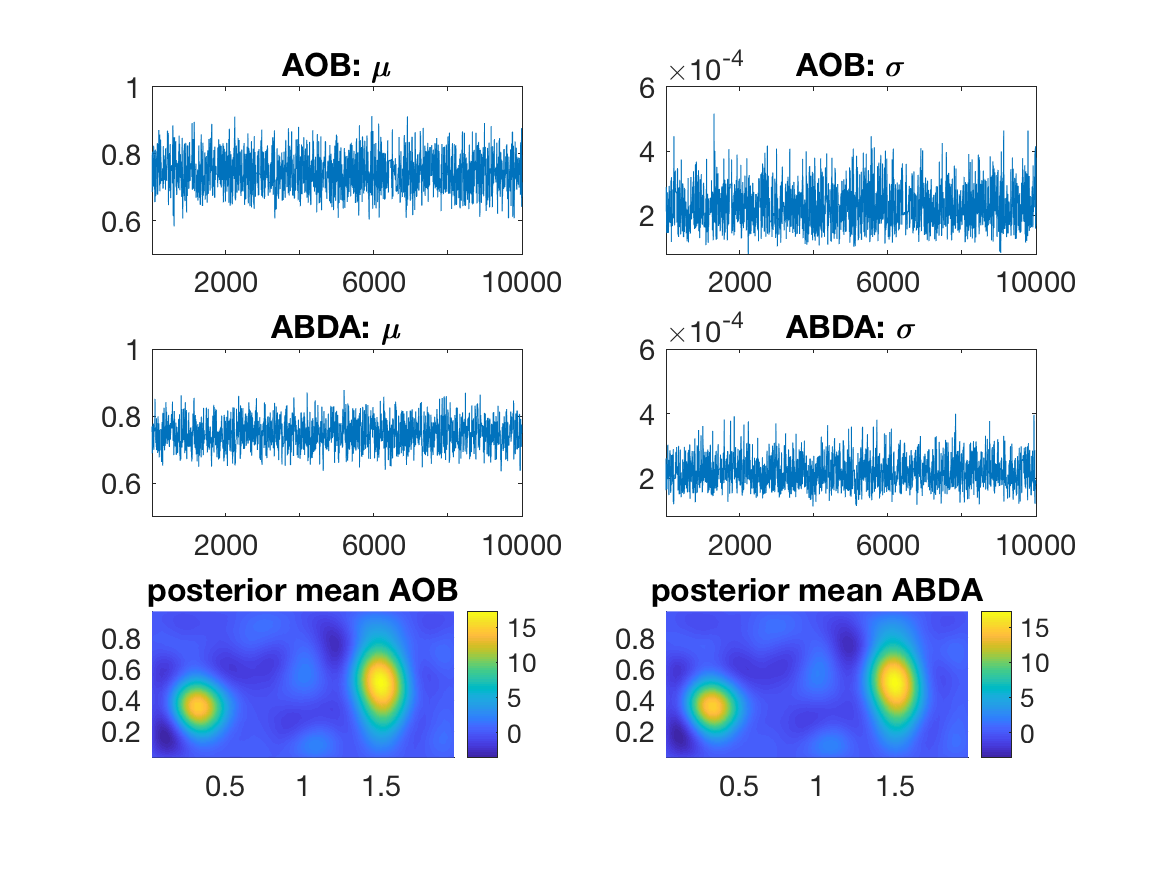}
\caption{Statistical inversion results for the
heat equation example, using AOB and ABDA methods. }
\label{fig:AOB_vs_ABDA}
\end{figure}

\begin{table}[!ht]
\begin{tabular}{l|ll}
metric & AOB & ABDA\\
\hline
time (sec) & 819.89    & 622.19 \\
ESS        & 1096.46   & 933.37\\
IACT       & 9.12 & 10.71
\end{tabular}
\caption{Sampling metrics for AOB and ABDA methods for
the inverse heat equation example.}
\label{tbl:AOB_vs_ABDA}
\end{table}

\section{Conclusions}\label{sec:conclusions}
This paper focuses on the problem of sampling from a posterior distribution that arises in hierarchical Bayesian inverse problems. We restrict ourselves to linear inverse problems with Gaussian measurement error, and we assume a Gaussian prior on the unknown quantity of interest. The hierarchical model arises when the precision parameters for both the measurement error and prior are assigned prior distributions. Gibbs sampling \cite{Bar11} and low-rank independence sampling \cite{brown2016computational} have been proposed for this type of problem, but the statistical efficiency of the resulting MCMC chains produced by either approach decreases as the dimension of the state increases. We discuss this phenomenon in detail and consider a solution using marginalization-based methods \cite{FoxNor,joyce2018point,RueHel}. Marginalization can be just as computationally prohibitive for sufficiently large-scale problems. Thus, we combine the low-rank techniques with the marginalization-based approaches to propose three MCMC algorithms that are computationally feasible for larger-scaled problems. The first of these new MCMC methods is a direct extension of the one-block algorithm \cite{RueHel}; the second is an extension of the delayed acceptance algorithm \cite{ChrFox}; and the last is an extension of the pseudo-marginal algorithm \cite{PseudoMarg}. We test and compare the performances of these three methods on two test cases, demonstrating that they all work reasonably well. We also offer suggestions on when one algorithm might be preferable over another based on, e.g., the spectrum of the prior-preconditioned Hessian matrix. {Future work will consider extensions to nonlinear forward models and non-Gaussian priors on the unknown quantity of interest.}

\section*{Acknowledgements}
This material was based upon work partially supported by the National Science Foundation under Grant DMS-1638521 to the Statistical and Applied Mathematical Sciences Institute.  The first author was also partially supported by NSF DMS-1720398. The third author was partially supported by NSF grants CMMI-1563435, EEC-1744497, and OIA-1826715. Any opinions, findings, and conclusions or recommendations expressed in this material are those of the author(s) and do not necessarily reflect the views of the National Science Foundation.
\bibliography{pcgibbs}

\begin{thebibliography}{10}

\bibitem{AgaBarPapStu}
S.~Agapiou, J.~M. Bardsley, O.~Papaspiliopoulos, and A.~M. Stuart.
\newblock Analysis of the {G}ibbs sampler for hierarchical inverse problems.
\newblock {\em SIAM/ASA J. Uncertain. Quantif.}, 2(1):511--544, 2014.

\bibitem{PseudoMarg}
C.~Andrieu and G.~O. Roberts.
\newblock The pseudo-marginal approach for efficient {Monte Carlo}
  computations.
\newblock {\em The Annals of Statistics}, pages 697--725, 2009.

\bibitem{Bar11}
J.~M. Bardsley.
\newblock {MCMC}-based image reconstruction with uncertainty quantification.
\newblock {\em SIAM J. Sci. Comput.}, 34(3):A1316--A1332, 2012.

\bibitem{bardsley2018}
J.~M. Bardsley.
\newblock {\em Computational Uncertainty Quantification for Inverse Problems}.
\newblock Society for Industrial and Applied Mathematics, 2018.

\bibitem{bardsley2013efficient}
J.~M. Bardsley, M.~Howard, and J.~G. Nagy.
\newblock Efficient {MCMC}-based image deblurring with {N}eumann boundary
  conditions.
\newblock {\em Electronic Transactions on Numerical Analysis}, 40:476--488,
  2013.

\bibitem{brooks2011handbook}
S.~Brooks, A.~Gelman, G.~Jones, and X.-L. Meng.
\newblock {\em Handbook of {Markov Chain Monte Carlo}}.
\newblock CRC press, 2011.

\bibitem{BrooksGelman98}
S.~P. Brooks and A.~Gelman.
\newblock General methods for monitoring convergence of iterative simulations.
\newblock {\em J. Comput. Graph. Stat.}, 7(4):434--455, 1998.

\bibitem{brown2016computational}
D.~A. Brown, A.~K. Saibaba, and S.~Vall{\'e}lian.
\newblock Low-rank independence samplers in hierarchical {B}ayesian inverse
  problems.
\newblock {\em SIAM/ASA Journal on Uncertainty Quantification},
  6(3):1076--1100, 2018.

\bibitem{Bui-ThanhGhattasMartinEtAl13}
T.~Bui-Thanh, O.~Ghattas, J.~Martin, and G.~Stadler.
\newblock A computational framework for infinite-dimensional {B}ayesian inverse
  problems {P}art {I}: {T}he linearized case, with application to global
  seismic inversion.
\newblock {\em SIAM Journal on Scientific Computing}, 35(6):A2494--A2523, 2013.

\bibitem{CalSom05}
D.~Calvetti and E.~Somersalo.
\newblock Priorconditioners for linear systems.
\newblock {\em Inverse problems}, 21:1397, 2005.

\bibitem{CalSom07}
D.~Calvetti and E.~Somersalo.
\newblock {\em {Introduction to Bayesian Scientific Computing}}.
\newblock Springer, 2007.

\bibitem{CalSom08}
D.~Calvetti and E.~Somersalo.
\newblock {Hypermodels in the Bayesian Imaging Framework}.
\newblock {\em Inverse Problems}, 24(3):034013, 2008.

\bibitem{ChrFox}
J.~A. Christen and C.~Fox.
\newblock {M}arkov chain {M}onte {C}arlo using an approximation.
\newblock {\em Journal of Computational and Graphical statistics},
  14(4):795--810, 2005.

\bibitem{Chung13}
J.~Chung and M.~Chung.
\newblock Computing optimal low-rank matrix approximations for image
  processing.
\newblock {\em Asilomar Conference on Signals, Systems and Computers}, page
  670674, 2013.

\bibitem{Chung14}
J.~Chung and M.~Chung.
\newblock An efficient approach for computing optimal low-rank regularized
  inverse matrices.
\newblock {\em Inverse Problems}, 30:114009, 2014.

\bibitem{cui2011bayesian}
T.~Cui, C.~Fox, and M.~J. O'Sullivan.
\newblock Bayesian calibration of a large-scale geothermal reservoir model by a
  new adaptive delayed acceptance {M}etropolis {H}astings algorithm.
\newblock {\em Water Resources Research}, 47(10), 2011.

\bibitem{CuiLawMar}
T.~Cui, K.~Law, and Y.~M. Marzouk.
\newblock Dimension-independent likelihood-informed {MCMC}.
\newblock {\em Journal of Computational Physics}, 304:109--137, 2016.

\bibitem{DaPrato06}
G.~Da~Prato.
\newblock {\em An introduction to infinite-dimensional analysis}.
\newblock Universitext. Springer-Verlag, Berlin, 2006.
\newblock Revised and extended from the 2001 original by Da Prato.

\bibitem{DaPratoZabczyk02}
G.~Da~Prato and J.~Zabczyk.
\newblock {\em Second order partial differential equations in {H}ilbert
  spaces}, volume 293 of {\em London Mathematical Society Lecture Note Series}.
\newblock Cambridge University Press, Cambridge, 2002.

\bibitem{DashtiStuart16}
M.~Dashti and A.~M. Stuart.
\newblock The {B}ayesian approach to inverse problems.
\newblock In R.~Ghanem, D.~Higdon, and H.~Owhadi, editors, {\em Handbook of
  Uncertainty Quantification}, pages 311--428. Springer International
  Publishing, Cham, 2017.

\bibitem{WilcoxGhattas}
H.~Flath, L.~Wilcox, V.~Akcelik, J.~Hill, B.~van Bloemen~Waanders, and
  O.~Ghattas.
\newblock Fast algorithms for {B}ayesian uncertainty quantification in
  large-scale linear inverse problems based on low-rank partial {H}essian
  approximations.
\newblock {\em SIAM Journal on Scientific Computing}, 33:407432, 2011.

\bibitem{FoxNor}
C.~Fox and R.~A. Norton.
\newblock Fast sampling in a linear-{G}aussian inverse problem.
\newblock {\em SIAM/ASA Journal on Uncertainty Quantification},
  4(1):1191--1218, 2016.

\bibitem{GelfandSmith90}
A.~E. Gelfand and A.~F.~M. Smith.
\newblock {Sampling-based approaches to calculating marginal densities}.
\newblock {\em Journal of the American Statistical Association},
  85(410):398--409, 1990.

\bibitem{Gelman06}
A.~Gelman.
\newblock {Prior distributions for variance parameters in hierarchical models}.
\newblock {\em Bayesian Analysis}, 1(3):515--533, 2006.

\bibitem{GelCarSteRub04}
A.~Gelman, J.~B. Carlin, H.~S. Stern, D.~B. Dunson, A.~Vehtari, and D.~B.
  Rubin.
\newblock {\em Bayesian data analysis}.
\newblock Texts in Statistical Science Series. CRC Press, Boca Raton, FL, third
  edition, 2014.

\bibitem{geman1993stochastic}
S.~Geman and D.~Geman.
\newblock Stochastic relaxation, {G}ibbs distributions and the {B}ayesian
  restoration of images.
\newblock {\em IEEE Transactions on Pattern Analysis and Machine Intelligence},
  6:721--741, 1984.

\bibitem{gilavert2015efficient}
C.~Gilavert, S.~Moussaoui, and J.~Idier.
\newblock Efficient {G}aussian sampling for solving large-scale inverse
  problems using {MCMC}.
\newblock {\em IEEE Transactions on Signal Processing}, 63(1):70--80, 2015.

\bibitem{haario2001adaptive}
H.~Haario, E.~Saksman, and J.~Tamminen.
\newblock An adaptive {M}etropolis algorithm.
\newblock {\em Bernoulli}, 7(2):223--242, 2001.

\bibitem{halko2011finding}
N.~Halko, P.-G. Martinsson, and J.~A. Tropp.
\newblock Finding structure with randomness: {P}robabilistic algorithms for
  constructing approximate matrix decompositions.
\newblock {\em SIAM review}, 53(2):217--288, 2011.

\bibitem{hastings1970monte}
W.~K. Hastings.
\newblock {Monte Carlo sampling methods using Markov chains} and their
  applications.
\newblock {\em Biometrika}, 57:97--109, 1970.

\bibitem{Higdon}
D.~Higdon.
\newblock {A primer on space-time modelling from a Bayesian perspective}.
\newblock {\em Statistical Methods for Spatio-Temporal Systems, eds.
  Finkenstadt, Held, and Isham}, pages 218--277, 2006.

\bibitem{IsaacPetraStadlerEtAl15}
T.~Isaac, N.~Petra, G.~Stadler, and O.~Ghattas.
\newblock Scalable and efficient algorithms for the propagation of uncertainty
  from data through inference to prediction for large-scale problems, with
  application to flow of the {A}ntarctic ice sheet.
\newblock {\em Journal of Computational Physics}, 296:348--368, September 2015.

\bibitem{joyce2018point}
K.~Joyce, J.~Bardsley, and A.~Luttman.
\newblock Point spread function estimation in {X}-ray imaging with partially
  collapsed {G}ibbs sampling.
\newblock {\em SIAM Journal on Scientific Computing}, 40(3):B766--B787, 2018.

\bibitem{KaiKolSomVau00}
J.~Kaipio, V.~Kolehmainen, E.~Somersalo, and M.~Vauhkonen.
\newblock {Statistical inversion and Monte Carlo sampling methods in electrical
  impedance tomography}.
\newblock {\em Inverse Problems}, 16(5):14871522, 2000.

\bibitem{KaiSom05}
J.~Kaipio and E.~Somersalo.
\newblock {\em {Statistical and Computational Inverse Problems}}.
\newblock Springer, New York, 2005.

\bibitem{KassEtAl98}
R.~E. Kass, B.~P. Carlin, A.~Gelman, and R.~Neal.
\newblock {Markov chain Monte Carlo in practice: A roundtable discussion}.
\newblock {\em Amer. Statist.}, 52:93--100, 1998.

\bibitem{LehmannCasella98}
E.~Lehmann and G.~Casella.
\newblock {\em {Theory of Point Estimation}}.
\newblock Springer, New York, 2nd edition, 1998.

\bibitem{Liu2004}
J.~S. Liu.
\newblock {\em {Monte Carlo Strategies in Scientific Computing}}.
\newblock Springer, New York, 2004.

\bibitem{MarWilBurGha}
J.~Martin, L.~C. Wilcox, C.~Burstedde, and O.~Ghattas.
\newblock {A stochastic Newton MCMC method for large-scale statistical Inverse
  Problems with application to seismic inversion}.
\newblock {\em SIAM Journal on Scientific Computing}, 34(3):A1460--A1487, 2012.

\bibitem{metropolis1953equation}
N.~Metropolis, A.~W. Rosenbluth, M.~N. Rosenbluth, A.~H. Teller, and E.~Teller.
\newblock Equation of state calculations by fast computing machines.
\newblock {\em The journal of chemical physics}, 21(6):1087--1092, 1953.

\bibitem{NicFox}
G.~Nicholls and C.~Fox.
\newblock {Prior modelling and posterior sampling in impedance imaging}.
\newblock {\em Bayesian Inference for Inverse Problems, Proc. SPIE},
  3459:116--127, 1998.

\bibitem{ParMar}
M.~Parno and Y.~M. Marzouk.
\newblock Transport map accelerated {Markov chain Monte Carlo}.
\newblock {\em SIAM/ASA Journal on Uncertainty Quantification}, 6(2):645--682,
  2018.

\bibitem{PetraMartinStadlerEtAl14}
N.~Petra, J.~Martin, G.~Stadler, and O.~Ghattas.
\newblock A computational framework for infinite-dimensional {B}ayesian inverse
  problems: {P}art {II}. {S}tochastic {N}ewton {MCMC} with application to ice
  sheet inverse problems.
\newblock {\em SIAM Journal on Scientific Computing}, 36(4):A1525--A1555, 2014.

\bibitem{RobertCasella04}
C.~Robert and G.~Casella.
\newblock {\em {Monte Carlo Statistical Methods}}.
\newblock Springer, New York, 2nd edition, 2004.

\bibitem{RueHel}
H.~Rue and L.~Held.
\newblock {\em Gaussian {M}arkov random fields}, volume 104 of {\em Monographs
  on Statistics and Applied Probability}.
\newblock Chapman \& Hall/CRC, Boca Raton, FL, 2005.
\newblock Theory and applications.

\bibitem{ScottBerger06}
J.~G. Scott and J.~O. Berger.
\newblock {An exploration of aspects of Bayesian multiple testing}.
\newblock {\em Journal of Statistical Planning and Inference}, 136:2144--2162,
  2006.

\bibitem{simon2000low}
H.~D. Simon and H.~Zha.
\newblock Low-rank matrix approximation using the {L}anczos bidiagonalization
  process with applications.
\newblock {\em SIAM Journal on Scientific Computing}, 21(6):2257--2274, 2000.

\bibitem{TC}
A.~Spantini, A.~Solonen, T.~Cui, J.~Martin, L.~Tenorio, and Y.~Marzouk.
\newblock Optimal low-rank approximations of {B}ayesian linear inverse
  problems.
\newblock {\em SIAM Journal on Scientific Computing}, 37:A2451--A2487, 2015.

\bibitem{Stuart10}
A.~M. Stuart.
\newblock Inverse problems: a {B}ayesian perspective.
\newblock {\em Acta Numerica}, 19:451--559, 2010.

\bibitem{tierney1994markov}
L.~Tierney.
\newblock {Markov chains} for exploring posterior distributions.
\newblock {\em The Annals of Statistics}, pages 1701--1728, 1994.

\end{thebibliography}
\bibliographystyle{abbrv}

\end{document}